\documentclass[journal,10pt]{IEEEtran}
\usepackage{amssymb}
\usepackage{amsmath}
\usepackage{cite}
\usepackage{url}
\usepackage{xcolor}
\usepackage{cite,graphicx,amsmath,amssymb}
\usepackage{subfigure}
\usepackage{citesort}
\usepackage{fancyhdr}
\usepackage{mdwmath}
\usepackage{mdwtab}
\usepackage{caption}
\usepackage{amsthm}
\usepackage{setspace}
\usepackage{algorithm}
\usepackage{algorithmic}
\usepackage{makecell}
\usepackage{diagbox}
\usepackage{stfloats}
\usepackage{balance} 
\usepackage{graphicx}
\usepackage{epstopdf}
\usepackage{epsfig}

\newtheorem{remark}{Remark}\newtheorem{theorem}{Theorem}

\newtheorem{lemma}{Lemma}

\newtheorem{corollary}{Corollary}

\allowdisplaybreaks
\makeatletter
\def\ScaleIfNeeded{
\ifdim\Gin@nat@width>\linewidth \linewidth \else \Gin@nat@width
\fi } \makeatother
 
\begin{document}
\title{Enhancing User Fairness in Wireless Powered Communication Networks with STAR-RIS}
\author{
	Guangyu~Zhu,~\IEEEmembership{Graduate Student Member,~IEEE,}
	Xidong~Mu,~\IEEEmembership{Member,~IEEE,}
	Li~Guo,~\IEEEmembership{Member,~IEEE,}
	Ao~Huang,~\IEEEmembership{Graduate Student Member,~IEEE,}
	Shibiao~Xu,~\IEEEmembership{Member,~IEEE}
	\thanks{An earlier version of this paper was presented in part at the IEEE Global Communications Conference Workshops (GC Wkshps), Kuala Lumpur, Malaysia, December 4-8, 2023 \cite{Zhu_conference}.}
	\thanks{Guangyu Zhu, Li Guo, Ao Huang and Shibiao Xu are with the Key Laboratory of Universal Wireless Communications, Ministry of Education, Beijing University of Posts and Telecommunications, Beijing 100876, China, also with the School of Artificial Intelligence, Beijing University of Posts and Telecommunications, Beijing 100876, China, and also with the National Engineering Research Center for Mobile Internet Security Technology, Beijing University of Posts and Telecommunications, Beijing 100876, China (email:\{Zhugy, guoli, huangao, shibiaoxu\}@bupt.edu.cn).}
	\thanks{Xidong Mu is with the Centre for Wireless Innovation (CWI), Queen's University Belfast, Belfast, BT3 9DT, U.K. (e-mail:
		x.mu@qub.ac.uk).}
}
\maketitle
\begin{abstract}
A simultaneously transmitting and reflecting reconfigurable intelligent surface (STAR-RIS) assisted wireless powered communication network (WPCN) is proposed, where two energy-limited devices first harvest energy from a hybrid access point (HAP) and then use that energy to transmit information back.
To fully eliminate the \emph{doubly-near-far} effect in WPCNs, two STAR-RIS operating protocol-driven transmission strategies, namely energy splitting non-orthogonal multiple access (ES-NOMA) and time switching time division multiple access (TS-TDMA) are proposed. For each strategy, the corresponding optimization problem is formulated to maximize the minimum throughput by jointly optimizing time allocation, user transmit power, active HAP beamforming, and passive STAR-RIS beamforming. For ES-NOMA, the resulting intractable problem is solved via a two-layer algorithm, which exploits the one-dimensional search and block coordinate descent methods in an iterative manner. For TS-TDMA, the optimal active beamforming and passive beamforming are first determined according to the maximum-ratio transmission beamformer. Then, the optimal solution of the time allocation variables is obtained by solving a standard convex problem. Numerical results show that: 1) the STAR-RIS can achieve considerable performance improvements for both strategies compared to the conventional RIS; 2) TS-TDMA is preferred for single-antenna scenarios, whereas ES-NOMA is better suited for multi-antenna scenarios; and 3) the superiority of ES-NOMA over TS-TDMA is enhanced as the number of STAR-RIS elements increases.
\end{abstract}

\section{Introduction}
Information handling services (IHS) Markit’s latest report predicts that the number of connected Internet-of-Things (IoT) devices will surge at an average annual rate of 12$\%$ and reach 125 billion in 2030 from 27 billion in 2017\cite{6G}. However, limited by the high cost and instability of currently used battery replacement or energy harvesting (EH)-based approaches, they will struggle to fulfill the sustainable communication requirements of massive energy-limited IoT devices in future networks \cite{Bi_WPCN}. To overcome this issue, a potential paradigm based on radio-frequency (RF) energy harvesting technology \cite{RF_energy}, known as the wireless powered communication network (WPCN), has been proposed in \cite{Zhang_WPCN,WPT_Survey}. According to \emph{harvest-then-transmit} (HTT) protocol \cite{Zhang_WPCN}, the entire communication process of a WPCN is divided into two phases, i.e., downlink wireless power transfer (WPT) and uplink wireless information transfer (WIT). In this regard, the energy-limited devices first harvest energy from the access point (AP) in the downlink and then use the harvested energy to transmit their data back in the uplink. Therefore, the rational time allocation between the two phases may facilitate a stable and controllable energy supply in communication networks. However, since both WPT and WIT are carried by RF signals, the system performance of WPCN is more susceptible to the transmission environment. Especially for long-distance transmission, the extremely low efficiency of WPT has become a major obstacle to the development and implementation of WPCN in practice.
In addition, having the same AP for power transmission and information reception will inevitably incur the \emph{doubly-near-far} phenomenon \cite{Zhang_WPCN}. In this case, the devices further away from the AP harvest less energy in the WPT but have to consume more power for the WIT, resulting in unfair performance between near devices and far devices. Accordingly, separating WPT and WIT execution on different APs becomes a promising solution to the unfairness problem \cite{Wu_double}. However, this design normally requires different antenna and RF systems with higher power consumption, design complexity, and hardware costs, which virtually poses new barriers to practical deployment. 

Recently, reconfigurable intelligent surfaces (RISs) \cite{RIS_survey} have emerged as an attractive solution to improve the transmission efficiency of future communication systems \cite{Huang_EE}. The RIS is a two-dimensional (2D) surface, comprising a large number of low-cost and passive meta-materials with tunable reflection properties. By dynamically adjusting the phase and amplitude of each element, the propagation environment can be reconstructed to enhance desirable signals and attenuate unwanted ones \cite{Wu_tutorial}. Inspired by this, the applications of RISs to various wireless networks, including WPCNs, have become a focus topic \cite{RIS_application1,Pan_application,WPCN_Application}. However, the feature of covering only \emph{half-space} variably weakens the effectiveness and flexibility of RIS deployment in practical implementation. To break this limitation, a new idea of RISs, known as simultaneously transmitting and reflecting RISs (STAR-RISs), has been suggested in \cite{Mu_star}. On the basis of conventional RIS reflection only, STAR-RISs can also transmit the incident signal to the opposite side of the RIS, thus achieving \emph{full-space} coverage \cite{Mu_survey}. Additionally, STAR-RISs has the capability to operate in three different modes: energy splitting (ES), time switching (TS), and mode switching (MS). Benefiting from this, more enhanced degrees of freedom (DoFs) and flexible communication resource allocation strategies can be exploited in system design, which renders the STAR-RIS a potential technique to enhance the performance of WPCNs from different aspects.

\subsection{Prior Works}
\emph{1) Studies on Conventional WPCNs}: Since the performance of WPCNs is not only dominated by uplink WIT but also governed by downlink WPT, balancing the two phases has become an essential issue in WPCNs. To this end, extensive research efforts have been devoted to resource allocation for WPCNs. Specifically, the authors of \cite{Zhang_WPCN} first introduced a simple WPCN employing time division multiple access (TDMA) for WIT, and then explored a time allocation strategy that maximizes the sum throughput. Following this, the time allocation optimization problem was further extended in \cite{WPCN_TDMA}, where a novel dynamic TDMA framework was applied in the WPCN. In addition to TDMA, the non-orthogonal multiple access (NOMA) strategy was also considered as a potential strategy for the uplink WIT of the WPCN in \cite{WPCN_NOMA}, where resource allocation was designed for the formulated sum throughput maximization problem. Moreover, the authors of \cite{Wu_Comparrsion} studied both TDMA-based and NOMA-based WPCNs together and revealed that NOMA is neither spectral efficient nor energy efficient for WPCN compared to TDMA. In \cite{Chi_Comparsion}, the authors investigated the energy provision minimization problem with the sum throughput requirement for both TDMA-based and NOMA-based WPCNs and came to similar conclusions as in \cite{Wu_Comparrsion}.
However, the above studies all focused only on the overall performance of the system and ignored the performance differences between individual devices. In fact, due to the emergence of the doubly-near-far effect, user fairness is more difficult to guarantee in WPCNs than in conventional communication networks. To address this challenge, in \cite{Zhang_WPCN}, the authors introduced a common throughput concept for the WPCN and then maximized the minimum throughput between two users by optimizing the time allocation of the TDMA strategy. The authors of \cite{Chen_CTM} studied a common throughput maximization problem for a NOMA-based WPCN, where the time allocation, transmit power of devices, and decoding order of devices' signals were jointly optimized. Furthermore, the authors of \cite{NOMA_fairness} suggested a novel transmission scheme with NOMA for the WPCN and showed that the scheme achieves a noteworthy enhancement in the rate for cell edge users compared to TDMA, thus improving system fairness.

2) \emph{Studies on RISs/STAR-RISs Assisted WPCNs}: Recently, the great potential advantages of using RISs/STAR-RISs for improving transmission efficiency have sparked a surge of interest in RISs/STAR-RISs assisted WPCNs. For instance, the authors of \cite{Xu_WPCN} considered a RIS aided WPCN where the power station and receive station are deployed separately, for which a radio resource and passive beamforming can be jointly optimized to improve system energy efficiency. In contrast, a RIS was introduced to support the downlink WPT and uplink WIT between a hybrid access point (HAP) and multiple devices in \cite{Wu_dynamic}. Moreover, the authors proposed three dynamic beamforming design frameworks for the sum throughput maximization problem. Driven by the same goal, the authors in \cite{Zheng_WPCN} studied a more complicated RIS assisted multiple input single output WPCN, where the optimization goal was achieved by jointly optimizing the active HAP beamforming, passive RIS beamforming as well as time allocation. Further considering a non-linear EH model, authors of \cite{Hua_MIMO} jointly optimized the active/passive beamforming and downlink/uplink time allocation to maximize the weighted sum throughput of a RIS assisted multiple input multiple output full-duplex WPCN. In addition, the authors of \cite{Self-Sustainable} adopted a realistic-based power consumption RIS model and investigated the sum throughput maximization problem in a self-sustainable RIS enabled WPCN. Different from the above TDMA-based WPCN, the sum throughput maximization problem for a RIS assisted NOMA-based WPCN was also examined in \cite{Wu_NOMA,Song_NOMA}. Particularly, the authors of \cite{Wu_NOMA} confirmed that only one RIS beamforming is needed in both downlink WPT and uplink WIT. On this basis, the authors of \cite{Zhang_hybrid_NOMA} proposed a hybrid-NOMA scheme for the uplink WIT to balance the beamforming complexity of the RIS and the sum throughput performance of the system. 
However, as mentioned above, the conventional RIS can only cover half-space for the WPCN. To break this limitation, the STAR-RIS was introduced into WPCNs in \cite{Du_STAR_WPCNs,Xie_STAR_WPCNs,Qin_STAR_MEC}. Among them, the authors investigated the sum throughput maximization problem for a STAR-RIS assisted WPCN in \cite{Du_STAR_WPCNs} and adopted TS working mode throughout the communication period. In \cite{Xie_STAR_WPCNs}, a STAR-RIS aided wireless powered NOMA system was proposed where ES and TS were considered for the downlink WPT while only ES was employed for the uplink WIT. On the contrary, the authors of \cite{Qin_STAR_MEC} applied the ES protocol in the downlink WPT and all three protocols in the uplink WIT for a STAR-RIS assisted wireless powered mobile edge computing (MEC) system. The results showed that TS is the best-performing protocol in terms of achieving uplink sum throughput.
\subsection{Motivations and Contributions}
Due to the doubly-near-far effect, addressing user fairness in WPCNs is a more urgent and challenging issue compared to conventional communication networks. Although RISs with channel reconfiguration capabilities are regarded as potential solutions to this problem, existing research on RISs in WPCNs remains focused on improving overall system performance rather than considering user fairness. This leaves a significant gap in the application of RISs for enhancing fairness in WPCNs. Moreover, compared to conventional RISs, STAR-RISs offer broader coverage capabilities and flexible communication protocols, providing more design freedom in geographic deployment and communication strategies to ensure user fairness. Therefore, utilizing STAR-RISs to enhance user fairness in WPCNs is an interesting and worthwhile topic to explore. However, to the best of our knowledge, there has been no related effort in this area, which motivates our work. As shown in Table I, we summarize the comparison between our work and existing works. It is evident that our study of STAR-RIS assisted WPCNs not only focuses on user fairness as the primary performance metric but also considers a more complex multi-antenna HAP configuration to fully leverage the joint beamforming gains. This complicated consideration is rarely investigated in existing WPCNs.
\begin{table*}[t]
\centering
\fontsize{9.5pt}{10pt}\selectfont
\caption{\textcolor{black}{The comparsion of the previous works and our work}}
\label{T1}
\renewcommand{\arraystretch}{1.5}  
\begin{tabular}{|c|m{1.3cm}<{\centering}|m{1.3cm}<{\centering}|m{1.3cm}<{\centering}|m{1.3cm}<{\centering}|m{1.3cm}<{\centering}|m{1.3cm}<{\centering}|m{1.3cm}<{\centering}|}
	\hline
	& \cite{Chen_CTM,NOMA_fairness} &\cite{Xu_WPCN,Wu_dynamic}&\cite{Zheng_WPCN,Hua_MIMO}&\cite{Wu_NOMA,Song_NOMA}&\cite{Du_STAR_WPCNs,Qin_STAR_MEC}&\cite{Xie_STAR_WPCNs}&\textbf{Our work}\\
	\hline
	{STAR-RIS enabled communication} &$\times$&$\times$&$\times$&$\times$&$\checkmark$&$\checkmark$ &$\checkmark$\\
	\hline
	{Single-antenna HAP} &$\checkmark$&$\checkmark$&$\checkmark$&$\checkmark$&$\checkmark$&$\checkmark$&$\checkmark$\\
	\hline
	{Multi-antenna HAP} &$\times$&$\times$&$\checkmark$&$\times$&$\times$&$\times$&$\checkmark$ \\
	\hline
	{TDMA transmission strategy} &$\times$&$\checkmark$&$\checkmark$&$\times$&$\checkmark$&$\times$&$\checkmark$ \\
	\hline
	{NOMA transmission strategy} &$\checkmark$&$\times$&$\times$&$\checkmark$&$\times$&$\checkmark$&$\checkmark$ \\
	\hline
	{User fairness as a performance metric} &$\checkmark$&$\times$&$\times$&$\times$&$\times$&$\times$&$\checkmark$ \\
	\hline
\end{tabular}
\end{table*}

However, this complex system also introduces new design challenges. Firstly, compared to single-antenna systems, a multi-antenna HAP requires additional transmission and reception beamforming designs to fully leverage its multiplexing gain. Secondly, to accommodate both signal transmission and reflection, STAR-RISs necessitate considering more variables than conventional RISs in the beamforming design. Moreover, these variables are highly-coupled, thus significantly increasing the design complexity. Thirdly, with the incorporation of STAR-RISs, TDMA is enhanced by pairing with TS because both focus on time-based resource allocation, improving interference-free communication efficiency. Conversely, NOMA is enhanced by pairing with ES, as both maximize resource utilization. Therefore, determining ``which communication strategy is better for WPCNs?" becomes an open question to be further explored. To overcome these challenges and obtain an upper bound on user fairness performance, in this paper, we investigate the minimum throughput maximization (MTM) problem for STAR-RIS assisted WPCNs. Specifically, two novel transmission strategies, namely ES-NOMA and TS-TDMA, are proposed. For each strategy, the formulated problem is solved by our proposed efficient algorithms. Based on this, the max-min user throughput is characterized and compared between the two communication strategies to answer the above question. It is suggested that \emph{TS-TDMA and ES-NOMA are the preferred transmission strategies for enhancing user fairness in STAR-RIS assisted WPCNs when the HAP is equipped with single and multiple antennas, respectively.}

It is important to note that, compared to our previous study in \cite{Zhu_conference}, we extend the assumption of a single antenna at the HAP to a multi-antenna configuration, which introduces more flexibility but also new complexities in joint beamforming design. Besides, our STAR-RIS protocol considerations now include TS in addition to ES, increasing communication scheme diversity and facilitating further exploration of STAR-RIS's potential to enhance user fairness in WPCNs. The main contributions of this paper are summarized as follows:
\begin{itemize}
\item We propose a STAR-RIS assisted WPCN, where a STAR-RIS is deployed to support both downlink WPT and uplink WIT between a multi-antenna HAP and two single-antenna users. To effectively eliminate the doubly-near-far effect, we delve into the user fairness performance in the WPCN based on two STAR-RIS operating protocol-driven transmission strategies, namely ES-NOMA and TS-TDMA. In the ES-NOMA strategy, the HAP first transmits energy to the users in the form of broadcasting, and then the users utilize the harvested energy to simulatenously transmit information to the HAP. While in the TS-TDMA strategy, the HAP utilizes orthogonal time slots to minimize interference, whether it is transmitting energy to different users or receiving information from different users. As a result, we formulate the MTM problem based on both ES-NOMA and TS-TDMA strategies.
\item For ES-NOMA, we propose a two-layer iterative algorithm to solve the resulting intractable problem. In the inner layer, we first determine the time allocation for WPT and WIT and then apply the block coordinate descent (BCD) framework to iteratively optimize highly-coupled variables. Particularly, we utilize the semidefinite relaxation (SDR) method and the minimum mean squared error (MMSE) criterion to design the active HAP beamforming for downlink WPT and uplink WIT, respectively. In addition, we apply the penalty-based method to optimize the passive STAR-RIS beamforming. While in the outer layer, we utilize the one-dimensional search to determine the optimal time allocation.
\item For TS-TDMA, we first employ the maximum-ratio transmission (MRT) beamformer to determine the optimal active beamforming and passive beamforming for both downlink WPT and uplink WIT. We then use convex optimization techniques to solve a standard convex problem and obtain the optimal time allocation solution.
\item Our numerical results depict that 1) the deployment of STAR-RISs is able to achieve significant gains for user fairness over conventional RISs in both NOMA-based WPCN and TDMA-based WPCN; 2) when the HAP is equipped with only one antenna, TS-TDMA achieves better performance, but ES-NOMA is more superior as the number of HAP antennas or STAR-RIS elements increases; and 3) deploying STAR-RISs closer to far users is more effective in mitigating the effects of doubly-near-far in WPCNs. 
\end{itemize}

\subsection{Organization and Notations}
The rest of the paper is structured as follows: Section II presents an introduction to the system model and the MTM frameworks for the ES-NOMA and TS-TDMA strategies. In Section III, effective algorithms are explored to solve the resulting intractable problems. Subsequently, Section IV shows the numerical results and the corresponding discussions. 
Finally, the paper concludes with Section V.

\emph{Notations}: Scalars, vectors, and matrices are denoted by lower-case, bold lower-case letters, and bold upper-case letters, respectively. $(\cdot)^H$ denotes the conjugate transpose. $\|\cdot\|$, $\|\cdot\|_2$, $\|\cdot\|_*$, and $\|\cdot\|_F$ denote the norm, spectral norm, nuclear norm, and Frobenius norm, respectively. $\mathrm{Tr}(\cdot)$ and $\mathrm{Rank}(\cdot)$ denote the trace and rank of the matrices. $\mathrm{diag}(\cdot)$ denotes the diagonalization operation on vectors. $\mathbb{C}^{M \times N}$ denotes the space of $M \times N$ complex valued matrices. $\mathbf{I}^{M\times M}$ denotes the unit matrix of order $M$. $\mathcal{CN}(\mu, \sigma^2)$ represents the distribution of a circularly symmetrical complex Gaussian random variable with a mean of $\mu$ and a variance of $\sigma^2$. 
\section{System Model and Problem Formulation}
\begin{figure}
\setlength{\abovecaptionskip}{0cm}   
\setlength{\belowcaptionskip}{0cm}   
\setlength{\textfloatsep}{7pt}
\centering
\includegraphics[width=3.6in]{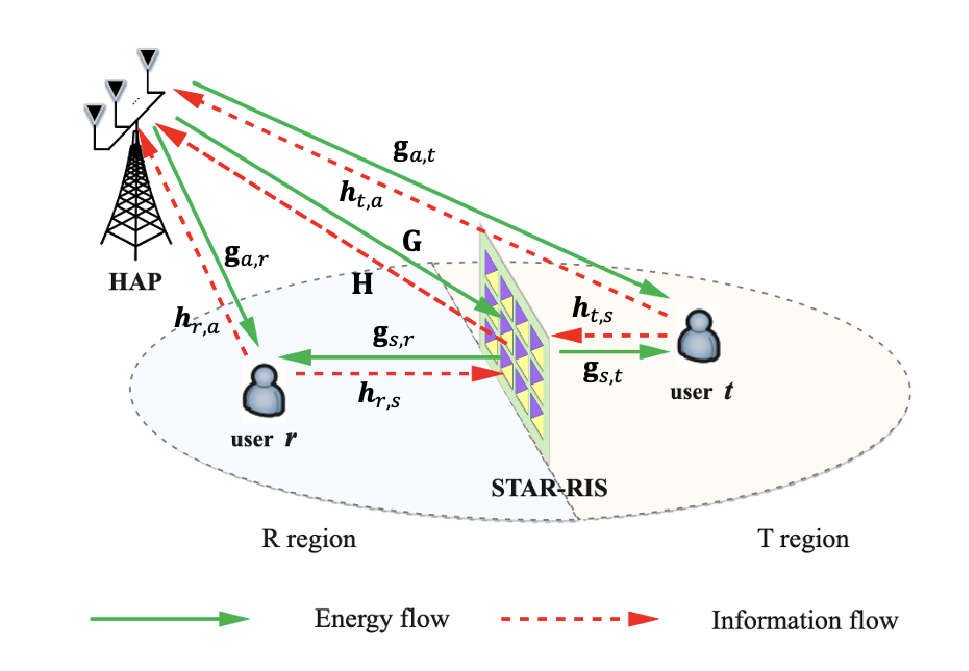}
\caption{Illustration of a SATR-RIS assisted WPCN.}
\label{fig:graph}
\end{figure}
We consider a STAR-RIS assisted WPCN, where a STAR-RIS consisting of $M$ elements is deployed to assist the downlink WPT and uplink WIT between a $N$-antenna HAP and multiple single-antenna users. In practice, these users are energy-limited IoT devices that have to first harvest energy from the HAP in order to supply subsequent information uploads. With the introduction of the STAR-RIS, the full communication space is divided into the transmission (T) and the reflection (R) regions. Among them, users in the R region are located on the same side of the STAR-RIS as the HAP, thus tending to have better channel conditions compared to users in the T region. Following this, to gain fundamental design insights from the STAR-RIS for eliminating the doubly-near-far effect in WPCNs, in the remainder of the paper we consider a basic setup, i.e., there is only one user $t$ and one user $r$ in the T region and R region, respectively, as depicted in Fig. 1. Moreover, we assume all channels follow the narrow-band quasi-static fading model. The coefficients of downlink WPT channels from the HAP to the STAR-RIS, from the HAP to user $t$, from the HAP to user $r$, from the STAR-RIS to user $t$, and from the STAR-RIS to user $r$ are denoted by $\mathbf{G}\in \mathbb{C}^{M\times N}$, $\mathbf{g}_{a,t}\in \mathbb{C}^{N\times 1}$, $\mathbf{g}_{a,r}\in \mathbb{C}^{N\times 1}$, $\mathbf{g}_{s,t}\in\mathbb{C}^{M\times 1}$, and $\mathbf{g}_{s,r}\in\mathbb{C}^{M\times 1}$, respectively. In contrast, the counterpart uplink WIT channel coefficients are denoted by $\mathbf{H}\in \mathbb{C}^{N\times M}$, $\mathbf{h}_{t,a}\in \mathbb{C}^{1\times N}$, $\mathbf{h}_{r,a}\in \mathbb{C}^{1\times N}$, $\mathbf{h}_{t,s}\in\mathbb{C}^{1\times M}$, and $\mathbf{h}_{r,s}\in\mathbb{C}^{1\times M}$, respectively. Consider that the channel estimation methods proposed in \cite{Xu_estimation} and \cite{Wu_estimation} can be utilized for efficient estimation of the downlink WPT channels and uplink WIT channels, respectively. Therefore, to investigate the maximum performance potential of the STAR-RIS assisted WPCNs, we make the assumption that the channel state information of all links is perfectly known at the HAP \cite{Wu_dynamic}.
\subsection{STAR-RIS Protocols and Models}
In this work, the STAR-RIS ES and TS protocols are applied to NOMA-based WPCN and TDMA-based WPCN, respectively.

For the ES protocol, the incident signal is divided into transmitted and reflected signals in the form of energy splitting. The amplitude adjustments of the $m$-th element for transmission and reflection are denoted as $\beta^t_m$ and $\beta^r_m$. Due to the law of conservation of energy, the constraints $\beta^t_m, \beta^r_m \in [0,1]$ and $\beta^t_m+\beta^r_m=1, \forall m \in\mathcal{M}=\{1,2,\cdots,M\}$ always hold. In addition, for each element of the STAR-RIS, there are two independent phase shifts for transmission and reflection, which are denoted by $\theta^t_m$ and $\theta^r_m$, respectively. Similar to conventional RISs, we assume the design of the phase shifts is continuous, i.e., $\theta^t_m, \theta^r_m \in [0,2\pi), \forall m\in\mathcal{M}.$ Thus, the transmission- and reflection-coefficient matrices of the STAR-RIS can be given by $\mathbf{\Theta}^{\textup{ES}}_t=\textup{diag}\left(\sqrt{\beta_1^{t}}e^{j\theta^t_1},\cdots,\sqrt{\beta_M^{t}}e^{j\theta^t_M}\right)$ and $\mathbf{\Theta}^{\textup{ES}}_r=\textup{diag}\left(\sqrt{\beta_1^{r}}e^{j\theta^r_1},\cdots,\sqrt{\beta_M^{r}}e^{j\theta^r_M}\right)$, respectively.

For the TS protocol, all elements of the STAR-RIS operate simultaneously in T mode or R mode at different orthogonal time durations to achieve transmission or reflection of the incident signal. Thus, the transmission- and reflection-coefficient matrices of the STAR-RIS are similar to the conventional RISs, and given by $\mathbf{\Theta}_t^{\textup{TS}}=\textup{diag}\left(e^{j\theta^t_1},\cdots,e^{j\theta^t_M}\right)$ and $\mathbf{\Theta}_r^{\textup{TS}}=\textup{diag}\left(e^{j\theta^r_1},\cdots,e^{j\theta^r_M}\right)$, respectively. In addition, let $\lambda_t$ and $\lambda_r$ denote the time allocation for T and R modes, respectively. The constraints $\lambda_t, \lambda_r \in[0,1]$ and $\lambda_t+\lambda_r=1$ hold throughout the communication period.
\subsection{Signal Transmission Models}
\begin{figure}
\setlength{\abovecaptionskip}{0cm}   
\setlength{\belowcaptionskip}{0.2cm}   
\setlength{\textfloatsep}{7pt}
\centering
\subfigure[ES-NOMA strategy]{
	\includegraphics[width=3.3in]{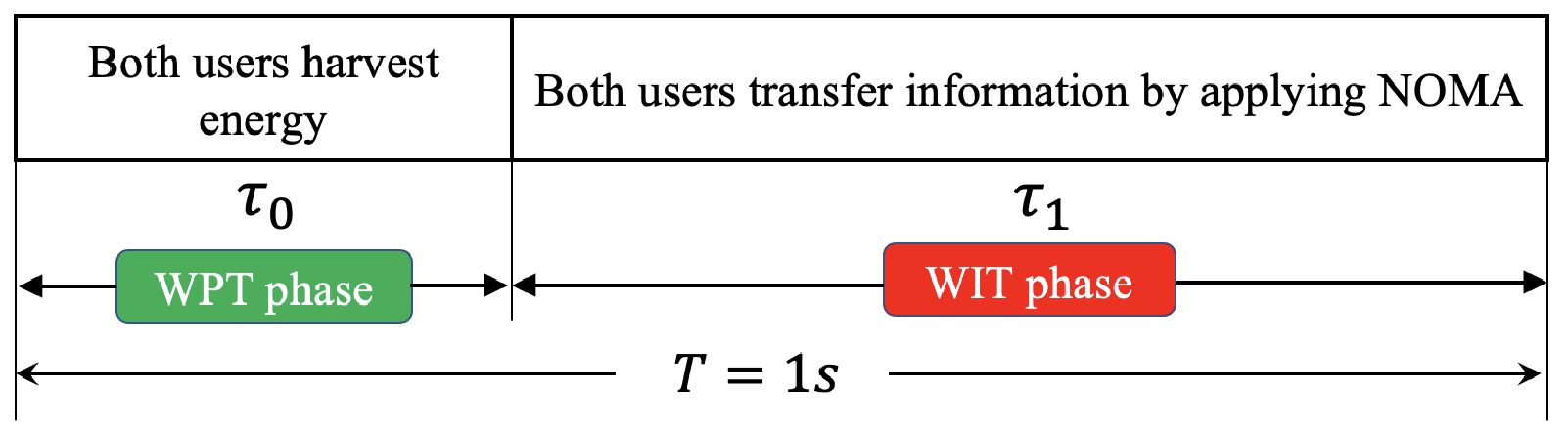}}
\subfigure[TS-TDMA strategy]{
	\includegraphics[width=3.3in]{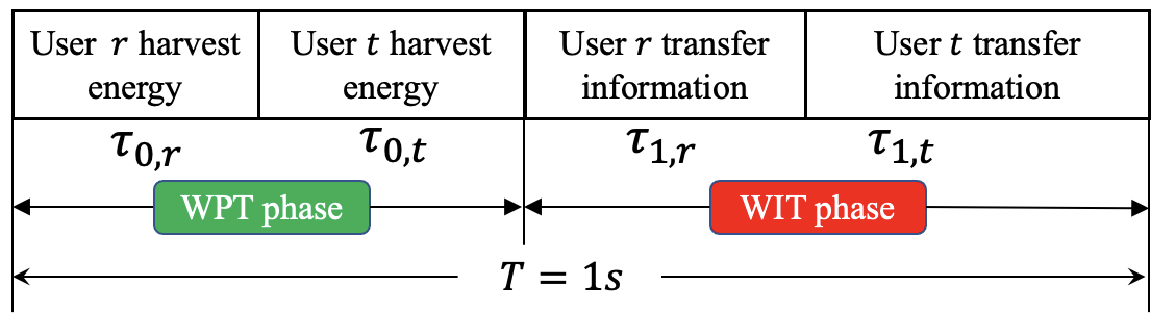}}
\caption{Illustration of transmission block structure.}
\end{figure}
According to the HTT protocol, the mission period $T$\footnote{In this paper, we consider a normalized unit transmission block time in the sequel, i.e., $T = 1 s$, to measure the system performance. Thus, the terms “power” and “energy” are interchangeable.} of each transmission block in WPCN contains both WIT and WPT phases. In the first phase, users harvest energy from the HAP using the allocated time. In the second phase, users will rely on the harvested energy to complete the information upload within the allocated time. In this paper, we investigate two specific transmission strategies, namely ES-NOMA and TS-TDMA, for the system. Furthermore, an uplink-adaptive dynamic beamforming framework is considered in the STAR-RIS \cite{Wu_dynamic}. To facilitate distinction, we denote the STAR-RIS beamforming matrices in the WPT and WIT phases as $\mathbf{\Theta}_k^{\textup{ES/TS}(1)}$ and $\mathbf{\Theta}_k^{\textup{ES/TS}(2)}, k\in\{t,r\}$, respectively. The following are explained in details:

\emph{1) Transmission Block Structure for ES-NOMA}: As illustrated in Fig. 2(a) for the ES-NOMA strategy, user $t$ and user $r$ simultaneously harvest energy from the HAP during the allocated time $\tau_0$. Let $\mathbf{v}^{\textup{ES}}_t\in \mathbb{C}^{N\times 1}$ and $\mathbf{v}^{\textup{ES}}_r\in \mathbb{C}^{N\times 1}$ denote the dedicated energy beams assigned for user $t$ and user $r$, respectively. Considering the transmit power budget $P_A$ at the HAP, we can have
\begin{align}
\|\mathbf{v}^{\textup{ES}}_t\|^2+\|\mathbf{v}^{\textup{ES}}_r\|^2\leq P_A.
\end{align} 
Since the noise power is considerably lower than the power received from the HAP, we assume that the interference and noise at users can be negligible \cite{Zhu_SWIPT}. Besides, the energy harvesting levels considered in this system can be completely covered by the linear conversion region generated by multi-parallel EH circuits \cite{Ma_NL}. Therefore, following the linear EH model \cite{Chen_MEC,XU_large},\footnote{By adopting a nonlinear EH model, our optimization framework remains applicable, but it requires further modifications. These adjustments will be the focus of our future work.
} the received RF energy at user $k$ is denoted as $E^{\textup{ES}}_k$, shown at the top of the next page, where $\eta$ indicates energy conversion efficiency.
\begin{figure*}[t]
	\normalsize
	\begin{align}	
		E^{\textup{ES}}_t=\tau_0\eta\left(|\mathbf{g}^H_{s,t}\mathbf{\Theta}^{\textup{ES(1)}}_t\mathbf{G}\mathbf{v}^{\textup{ES}}_t+\mathbf{g}^H_{a,t}\mathbf{v}^{\textup{ES}}_t|^2+|\mathbf{g}^H_{s,t}\mathbf{\Theta}^{\textup{ES(1)}}_t\mathbf{G}\mathbf{v}^{\textup{\textup{ES}}}_t+\mathbf{g}^H_{a,t}\mathbf{v}^{\textup{ES}}_t|^2\right),\\
		E^{\textup{ES}}_r=\tau_0\eta\left(|\mathbf{g}^H_{s,r}\mathbf{\Theta}^{\textup{ES(1)}}_r\mathbf{G}\mathbf{v}^{\textup{ES}}_r+\mathbf{g}^H_{a,r}\mathbf{v}^{\textup{ES}}_r|^2+|\mathbf{g}^H_{s,r}\mathbf{\Theta}^{\textup{ES(1)}}_r\mathbf{G}\mathbf{v}^{\textup{ES}}_t+\mathbf{g}^H_{a,r}\mathbf{v}^{\textup{ES}}_t|^2\right).
	\end{align}
	\hrulefill\vspace*{0pt}
\end{figure*}
\begin{remark}
	\textup{For the ES-NOMA strategy, since the energy beams do not carry any information and STAR-RIS may split their energy to both T and R regions, all energy beams will be received as broadcasts by user $t$ and user $r$ simultaneously. In this case, only one energy beam is needed. Therefore, instead of $\mathbf{v}^{\textup{ES}}_t$ and $\mathbf{v}^{\textup{ES}}_t$, we will consider $\mathbf{v}^{\textup{ES}}$ (where $\mathbf{v}^{\textup{ES}}=\mathbf{v}^{\textup{ES}}_t+\mathbf{v}^{\textup{ES}}_r$) as the combined energy beam in the subsequent study.}
\end{remark}
Let $\mathbf{G}_k=[\mathrm{diag}(\mathbf{g}^H_{s,k})\mathbf{G};\mathbf{g}^H_{a,k}] \in \mathbb{C}^{(M+1)\times N}, k\in\{t,r\}$ denote the combined channel from the HAP to user $k$. $\widetilde{\mathbf{u}}^{\textup{ES}}_k=[\mathbf{u}^{\textup{ES}}_k;1]\in\mathbb{C}^{(M+1)\times1}$, where $\mathbf{u}^{\textup{ES}}_{k}=[\sqrt{\beta_1^{k(1)}}e^{j\theta^{k(1)}_1},\!\cdots\!,\sqrt{\beta_M^{k(1)}}e^{j\theta^{k(1)}_M}]^T, {k}\in\{t,r\}$. Then, the harvested energy for user $k$ can be further written as
\begin{align}
	E^{\textup{ES}}_k=\tau_0\eta|(\widetilde{\mathbf{u}}^{\textup{ES}}_k)^H\!\mathbf{G}_k\mathbf{v}^{\textup{ES}}|^2,k\in\{t,r\}.
\end{align}

After energy harvesting, power-domain NOMA is adopted for the uplink WIT. Accordingly, throughout the duration of $\tau_1$, user $t$ and user $r$ simultaneously transmit their respective information to the HAP at a transmit power $p_k$ of $\tau_1 p_k \leq E^{\textup{ES}}_k, k\in\{t,r\}$. As a result, the received signal at the HAP is given by
\begin{align}
	\mathbf{y}=\sum_{k\in\{t,r\}}\left(\mathbf{h}^H_{k,a}+\mathbf{H}\mathbf{\Theta}^{\textup{ES}(2)}_k\mathbf{h}^H_{k,s}\right)\sqrt{p_k}x_k+\mathbf{n}_{0},
\end{align}
where $x_k$ denotes the transmit signal of user $k$ with $x_k \sim \mathcal{CN}(0,1)$. $\mathbf{n}_0 \sim \mathcal{CN}(0,\sigma^2\mathbf{I})$ denotes the additive white Gaussian noise at the HAP. Let $\mathbf{H}_k=[\mathbf{H}\mathrm{diag}(\mathbf{h}^H_{k,s}),\mathbf{h}^H_{k,a}]\in\mathbb{C}^{N\times(M+1)}, k\in\{t,r\}$ denote the combined channels from user $k$ to the HAP, and $\widetilde{\mathbf{q}}^{\textup{ES}}_k=[\mathbf{q}^{\textup{ES}}_k;1]\in\mathbb{C}^{(M+1)\times 1}$, where $\mathbf{q}^{\textup{ES}}_{k}=[\sqrt{\beta_1^{k(2)}}e^{j\theta^{k(2)}_1},\cdots,\sqrt{\beta_M^{k(2)}}e^{j\theta^{k(2)}_M}]^T, {k}\in\{t,r\}$. As a result, the received signal is further expressed as 
\begin{align}
	\mathbf{y}=\mathbf{H}_t\widetilde{\mathbf{q}}^{\textup{ES}}_t\sqrt{p_t}x_t+\mathbf{H}_r\widetilde{\mathbf{q}}^{\textup{ES}}_r\sqrt{p_r}x_r+\mathbf{n}_0.
\end{align}
Here, we consider that a linear beamforming vector $\mathbf{w}^{\textup{ES}}_k \in \mathbb{C}^{N\times 1}$ is deployed to receive the superimposed signals \cite{Zhang_uplink}. Then, the estimated signal is shown as
\begin{align}
	r_k=\left(\mathbf{w}_k^{\textup{ES}}\right)^H\mathbf{y}, \forall k \in \{t,r\}.
\end{align} 

In addition, to mitigate inter-user interference, successive interference cancellation (SIC) is performed at the HAP. More particularly, this paper supposes that the HAP first decodes the message of user $r$, while the message of user $t$ is treated as noise.\footnote{Intuitively, user $r$ tends to be closer to the HAP, so it has higher direct link conditions. Decoding its message first helps improve system performance. Therefore, for the sake of simplicity, the investigation of decoding order is not the focus here but will be presented in future work.} Then the message of user $r$ will be subtracted from the received composite signal, and the message of user $t$ is decoded independently. As a consequence, the achievable data rate of user $k$ can be stated as
\begin{gather}
	R^{\textup{ES}}_t=\tau_1\log_2\left(1+\frac{p_t|(\mathbf{w}^{\textup{ES}}_t)^H\mathbf{H}_t\widetilde{\mathbf{q}}^{\textup{ES}}_t|^2}{\|\mathbf{w}^{\textup{ES}}_t\|^2\sigma^2}\right),\\
	R^{\textup{ES}}_r\!=\!\tau_1\log_2\left(1+\frac{p_r|(\mathbf{w}^{\textup{ES}}_r)^H\mathbf{H}_r\widetilde{\mathbf{q}}^{\textup{ES}}_r|^2}{p_t|(\mathbf{w}^{\textup{ES}}_r)^H\mathbf{H}_t\widetilde{\mathbf{q}}^{\textup{ES}}_t|^2+\|\mathbf{w}^{\textup{ES}}_r\|^2\sigma^2}\right).
\end{gather}

\emph{2) Transmission Block Structure for TS-TDMA}: As shown in Fig. 2(b), the whole transmission block for TS-TDMA is also divided into two parts, i.e., $\tau_{0,k}$ for WPT and $\tau_{1,k}$ for WIT, as in the ES-NOMA strategy. The difference is that due to the time switching characteristics of the STAR-RIS when employing the TS protocol, each time block is further allocated to user $t$ and user $r$ based on $\tau_{0,t}+\tau_{0,r}+\tau_{1,t}+\tau_{1,r}=T$. Then, by denoting the dedicated energy beams for user $t$ and user $r$ as $\mathbf{v}^{\textup{TS}}_t$ and $\mathbf{v}^{\textup{TS}}_r$, respectively, the harvested energy for user $k$ is expressed as
\begin{align}
	E^{\textup{TS}}_t=\tau_{0,t}\eta|(\widetilde{\mathbf{u}}^{\textup{TS}}_t)^H\mathbf{G}_t\mathbf{v}_t^{\textup{TS}}|^2, \\
	E^{\textup{TS}}_r=\tau_{0,r}\eta|(\widetilde{\mathbf{u}}^{\textup{TS}}_r)^H\mathbf{G}_r\mathbf{v}_r^{\textup{TS}}|^2,
\end{align}
where $\widetilde{\mathbf{u}}^{\textup{TS}}_k=[\mathbf{u}_k^{\textup{TS}};1]\in\mathbb{C}^{(M+1)\times 1}$,   $\mathbf{u}_{k}^{\textup{TS}}=[e^{j\theta^{k(1)}_1},\cdots,e^{j\theta^{k(1)}_M}]^T, {k}\in\{t,r\}$. It should be noted that although user $t$ can also harvest energy from $\mathbf{v}_r^{\textup{TS}}$ during $\tau_{0,r}$, it can only receive signal from the direct link since the STAR-RIS is operating in R mode during this period. Consequently, the harvested energy during this time is at a very low level compared to that harvested during $\tau_{0,t}$ for user $t$. Therefore, we ignore this part of the energy to reduce the design complexity of the energy beams as in \cite{Xie_STAR_WPCNs,Du_STAR_WPCNs}. The same applies to user $r$. Further, since TDMA is adopted for the uplink WIT, the HAP experiences no inter-user interference. As a consequence, the achievable data rate from user $k$ at the HAP is shown as 
\begin{align}
	R^{\textup{TS}}_t=\tau_{1,t}\log_2\left(1+\frac{p_t|(\mathbf{w}^{\textup{TS}}_t)^H\mathbf{H}_t\widetilde{\mathbf{q}}^{\textup{TS}}_t|^2}{\|\mathbf{w}^{\textup{TS}}_t\|^2\sigma^2}\right),\\
	R^{\textup{TS}}_r=\tau_{1,r}\log_2\left(1+\frac{p_r|(\mathbf{w}^{\textup{TS}}_r)^H\mathbf{H}_r\widetilde{\mathbf{q}}^{\textup{TS}}_r|^2}{\|\mathbf{w}^{\textup{TS}}_r\|^2\sigma^2}\right),
\end{align}
where $\widetilde{\mathbf{q}}^{\textup{TS}}_k=[\mathbf{q}_k^{\textup{TS}};1]\in\mathbb{C}^{(m+1)\times 1}$, $\mathbf{q}_{k}^{\textup{TS}}=[e^{j\theta^{k(2)}_1},\cdots,e^{j\theta^{k(2)}_M}]^T, {k}\in\{t,r\}$. 

\subsection{Problem Formulation}
To enhance user fairness, our goal is to investigate the MTM problem between user $t$ and user $r$ by jointly optimizing the time allocation, user transmit power, active transmit and receive beamforming at the HAP, as well as the downlink/uplink passive STAR-RIS beamforming. Taking into consideration the different transmission strategies, the corresponding MTM problems are formulated in the following manner:

\emph{1) The MTM Problem for ES-NOMA}:
\begin{subequations}\label{ES_CTM}
	\begin{align}
		&\label{P1_C0}\  \max_{{\{\mathbf{v^{\textup{ES}}}\},\{\mathbf{w}^{\textup{ES}}_k\},\{\mathbf{u}_{k}^{\textup{ES}}\},\{\mathbf{q}_{k}^{\textup{ES}}\},\{p_k\},\tau_0,\tau_1}} \min_{k\in\{t,r\}} R^{\textup{ES}}_k\\ 
		&\label{P1_C1} {\rm s.t.} \ \left\|\mathbf{v}^{\textup{ES}}\right\|^2 \leq P_A,\\
		&\label{P1_C2} \quad \ \ \left\|\mathbf{w}^{\textup{ES}}_k\right\|^2=1, k\in\{t,r\}, \\
		&\label{P1_C3} \quad \ \ \tau_0>0, \tau_1>0, p_k>0, k \in \{t,r\}, \\
		&\label{P1_C4} \quad \ \  \tau_0+\tau_1=T, \\
		&\label{P1_C5} \quad \ \ \tau_1 p_k \leq E^{\textup{ES}}_k, k \in \{t,r\},\\
		&\label{P1_C6} \quad \ \ \textcolor{black}{ \mathbf{u}^{\textup{ES}}_k\!=\![\sqrt{\beta_1^{k(1)}}e^{j\theta^{k(1)}_1},\cdots\!,\sqrt{\beta_M^{k(1)}}e^{j\theta^{k(1)}_M}]^T, k \in \{t,r\},}\\,
		&\label{P1_C7} \quad \ \ \textcolor{black}{ \mathbf{q}^{\textup{ES}}_{k}\!=\![\sqrt{\beta_1^{k(2)}}e^{j\theta^{k(2)}_1},\cdots\!,\sqrt{\beta_M^{k(2)}}e^{j\theta^{k(2)}_M}]^T,k \in \{t,r\},} \\
		&\label{P1_C8} \quad \ \ \theta^{k(i)}_m \in [0,2\pi),\forall m \in \mathcal{M}, k\in\{t,r\}, i\in\{1,2\},\\
		&\quad\ \ \beta_m^{t(i)}, \beta_m^{r(i)}\in [0,1],\beta_{m}^{t(i)}+\beta_{m}^{r(i)}=1, \forall m\in \mathcal{M}, \nonumber \\ 
		&\label{P1_C9}\quad\ \ i\in\{1,2\},
	\end{align}
\end{subequations}
where constraint \eqref{P1_C1} denotes the HAP transmit power budget in the downlink WPT, while \eqref{P1_C2} represents the normalization constraint for HAP receive beamforming in the uplink WIT. Constraints \eqref{P1_C3} and \eqref{P1_C4} ensure the feasibility of user transmit power and time allocation. Additionally, constraint \eqref{P1_C5} ensures the causality of energy harvesting and consumption between WPT and WIT. \eqref{P1_C6} and \eqref{P1_C7} represent the STAR-RIS tuning vectors in downlink WPT and uplink WIT, respectively. Moreover, \eqref{P1_C8} and \eqref{P1_C9} denote the phase shift and amplitude constraints on the STAR-RIS beamforming design in both transmission phases.

\emph{2) The MTM Problem for TS-TDMA}:
\begin{subequations}\label{TS_CTM}
	\begin{align}
		&\label{P2_C0}\  \max_{{\{\mathbf{v_k^{\textup{TS}}}\},\{\mathbf{w}^{\textup{TS}}_k\},\{\mathbf{u}_{s}^{\textup{TS}}\},\{\mathbf{q}_{s}^{\textup{TS}}\},\{p_k\},\{\tau_{0,k}\},\{\tau_{1,k}\}}} \min_{k\in\{t,r\}} R^{\textup{TS}}_k\\ 
		&\label{P2_C1} \ \ {\rm s.t.} \ \left\|\mathbf{v}_k^{\textup{TS}}\right\|^2\leq P_A, k\in\{t,r\}, \\
		&\label{P2_C2} \quad \quad \ \left\|\mathbf{w}^{\textup{TS}}_k\right\|^2=1, k\in\{t,r\}, \\
		&\label{P2_C3} \quad \quad \ \tau_{0,k}>0, \tau_{1,k}>0, p_k>0, k \in \{t,r\}, \\
		&\label{P2_C4} \quad \quad \  \tau_{0,t}+\tau_{0,r}+\tau_{1,t}+\tau_{1,r}=T, \\
		&\label{P2_C5} \quad \quad \ \tau_{1,k} p_k \leq E^{\textup{TS}}_k, k \in \{t,r\},\\
		&\label{P2_C6} \quad \quad \ \textcolor{black}{\mathbf{u}_{k}^{\textup{TS}}=[e^{j\theta^{k(1)}_1},\cdots,e^{j\theta^{k(1)}_M}]^H, {k}\in\{t,r\},}\\
		&\label{P2_C7} \quad \quad \ \textcolor{black}{ \mathbf{q}_{k}^{\textup{TS}}=[e^{j\theta^{k(2)}_1},\cdots,e^{j\theta^{k(2)}_M}]^H, {k}\in\{t,r\},}\\
		&\label{P2_C8} \quad \quad \ \theta^{k(i)}_m \in [0,2\pi),\forall m \in \mathcal{M}, k\in\{t,r\}, i\in\{1,2\}.
	\end{align}
\end{subequations}
Similarly, \eqref{P2_C1} and \eqref{P2_C2} denote the HAP beamforming constraints in the TS-TDMA communication strategy. \eqref{P2_C3}-\eqref{P2_C5} ensure the validity of the strategy. Besides, \eqref{P2_C6} and \eqref{P2_C7} represent the STAR-RIS tuning vectors in downlink WPT and uplink WIT. Unlike ES-NOMA, the full space coverage of STAR-RIS is realized in TS-TDMA by time switching. Thus, the STAR-RIS beamforming design is only limited by the phase shift constraint \eqref{P2_C8}.

As can be observed, since all optimization variables are tightly coupled in the objective functions and constraints, both problems $\textup{(\ref{ES_CTM})}$ and $\textup{(\ref{TS_CTM})}$ are highly non-convex in their current forms and are challenging to solve directly. To get around these obstacles, we shall examine effective algorithms for solving these intractable problems in the next section.

\section{Proposed Solutions to the MTM Problems}
In order to tackle the MTM problem based on two different transmission strategies, we first propose an iterative algorithm for the ES-NOMA strategy based on a two-layer framework. Then, for the TS-TDMA strategy, we employ the MRT beamformer to determine the optimal active beamforming and decompose the remaining variables into two blocks, i.e., STAR-RIS passive beamforming and time allocation, which can be solved in a relatively straightforward manner.
\begin{figure*}[t]
	\normalsize
	\begin{subequations}
		\begin{gather}
			\label{Rate1}R^{\textup{ES}}_t=(1-\overline{\tau})\log_2\left(1+\frac{p_t\mathrm{Tr}\left(\mathbf{W}^{\textup{ES}}_t\mathbf{H}_t\widetilde{\mathbf{Q}}^{\textup{ES}}_t\mathbf{H}^H_t\right)}{\sigma^2}\right)\\
			\label{Rate2}R^{\textup{ES}}_r=(1-\overline{\tau})\log_2\left(1+\frac{p_r\mathrm{Tr}\left(\mathbf{W}^{\textup{ES}}_r\mathbf{H}_r\widetilde{\mathbf{Q}}^{\textup{ES}}_r\mathbf{H}^H_r\right)}{p_t\mathrm{Tr}\left(\mathbf{W}^{\textup{ES}}_r\mathbf{H}_t\widetilde{\mathbf{Q}}^{\textup{ES}}_t\mathbf{H}^H_t\right)+\sigma^2}\right)
		\end{gather}
	\end{subequations}
	\hrulefill \vspace*{0pt}
\end{figure*}
\subsection{Proposed Solution for ES-NOMA}
To start with, we propose a two-layer framework to optimize all variables for ES-NOMA. In the inner layer, we initially establish time allocation for $\tau_0$ and $\tau_1$, followed by solving a simplified problem for $\textup{(\ref{ES_CTM})}$ to determine the optimal solution for the remaining variables. 
While in the outer layer, the optimal solutions for $\tau_0$ and $\tau_1$ are solved via one-dimensional search. Now, we focus on the inner layer design. For any given $\tau_0=\overline{\tau}$ and $\tau_1=1-\overline{\tau}$, we define $\mathbf{V}^{\textup{ES}}=\mathbf{v}^{\textup{ES}}(\mathbf{v}^{\textup{ES}})^H \in\mathbb{C}^{N\times N}$, $\mathbf{W}^{\textup{ES}}_k=\mathbf{w}^{\textup{ES}}_k(\mathbf{w}^{\textup{ES}}_k)^H \in\mathbb{C}^{N\times N}$, $\widetilde{\mathbf{Q}}^{\textup{ES}}_k\triangleq\widetilde{\mathbf{q}}^{\textup{ES}}_k(\widetilde{\mathbf{q}}^{\textup{ES}}_k)^H \in\mathbb{C}^{(M+1)\times (M+1)}$, and $\widetilde{\mathbf{U}}^{\textup{ES}}_k\triangleq\widetilde{\mathbf{u}}^{\textup{ES}}_k(\widetilde{\mathbf{u}}^{\textup{ES}}_k)^H\in\mathbb{C}^{(M+1)\times (M+1)}$. As such, the achievable data rate for user $t$ and user $r$ can be further expressed as \eqref{Rate1} and \eqref{Rate2}, presented at the top of the page.
Next, by introducing an auxiliary variable $\Gamma$ with $R^{\textup{ES}}_k \geq \Gamma, \forall k\in\{t,r\}$, problem $\textup{(\ref{ES_CTM})}$ is reformulated into an equivalent semi-definite program (SDP) form as follows:
\begin{subequations}\label{ES_SDP}
	\begin{align}
		&\label{P3_C0}\  \max_{{\{\mathbf{V}^{\textup{ES}}\},\{\mathbf{W}^{\textup{ES}}_k\},\{\widetilde{\mathbf{U}}_{k}^{\textup{ES}}\},\{\widetilde{\mathbf{Q}}_{k}^{\textup{ES}}\},\{p_k\}}} \Gamma \\
		&\label{P3_C1}\ {\rm s.t.} \ (1\!-\!\overline{\tau})\log_2\!\left(\!1+\frac{p_t\!\mathrm{Tr}\!\left(\!\mathbf{W}^{\textup{ES}}_t\mathbf{H}_t\widetilde{\mathbf{Q}}^{\textup{ES}}_t\mathbf{H}^H_t\!\right)}{\sigma^2}\!\right) \geq \Gamma,\\
		&\label{P3_C2}\quad \quad (1\!-\!\overline{\tau})\!\log_2\!\!\left(\!1\!+\!\frac{p_r\!\mathrm{Tr}\!\left(\!\mathbf{W}^{\textup{ES}}_r\mathbf{H}_r\widetilde{\mathbf{Q}}^{\textup{ES}}_r\mathbf{H}^H_r\!\right)}{p_t\!\mathrm{Tr}\!\left(\!\mathbf{W}^{\textup{ES}}_r\mathbf{H}_t\widetilde{\mathbf{Q}}^{\textup{ES}}_t\mathbf{H}^H_t\!\right)\!\!+\!\sigma^2}\!\right)\!\!\geq\! \! \Gamma,\\
		&\label{P3_C3}\quad \quad \mathrm{Tr}\left(\mathbf{V}^{\textup{ES}}\right) \leq P_A,\\
		&\label{P3_C4}\quad \quad(1-\overline{\tau})p_k\! \leq \! \overline{\tau}\mathrm{Tr}\left(\!\mathbf{V^{\textup{ES}}}\mathbf{G}^H_k\widetilde{\mathbf{U}}^{\textup{ES}}_k\mathbf{G}_k\!\right), k\!\in\!\{t,r\},\\
		&\label{P3_C5}\quad\quad \mathbf{V}^{\textup{ES}} \succeq0, \mathrm{Rank}\left(\mathbf{V}^{\textup{ES}}\right)=1, \\
		&\label{P3_C6}\quad \quad \mathbf{W}^{\textup{ES}}_k\succeq 0, \mathrm{Rank}\left(\mathbf{W}^{\textup{ES}}_k\right)=1, k\in\{t,r\}, \\
		&\label{P3_C7}\quad\quad  \widetilde{\mathbf{U}}^{\textup{ES}}_k\succeq 0, \mathrm{Rank}\left(\widetilde{\mathbf{U}}^{\textup{ES}}_k\right)=1, \forall k\in\{t,r\}, \\
		&\label{P3_C8}\quad\quad  \widetilde{\mathbf{Q}}^{\textup{ES}}_k\succeq 0, \mathrm{Rank}\left(\widetilde{\mathbf{Q}}^{\textup{ES}}_k\right)=1, \forall k\in\{t,r\}, \\
		&\label{P3_C9}\quad\quad [\widetilde{\mathbf{U}}^{\textup{ES}}_t]_{m,m}+[\widetilde{\mathbf{U}}^{\textup{ES}}_r]_{m,m}=1,\forall m\in\mathcal{M},\\ 
		&\label{P3_C10}\quad\quad [\widetilde{\mathbf{Q}}^{\textup{ES}}_t]_{m,m}+[\widetilde{\mathbf{Q}}^{\textup{ES}}_r]_{m,m}=1,\forall m\in\mathcal{M},\\ 
		&\label{P3_C11}\quad \quad [\widetilde{\mathbf{U}}^{\textup{ES}}_k]_{M+1,M+1}\!\!=\!\!1,[\widetilde{\mathbf{Q}}^{\textup{ES}}_k]_{M+1,M+1}\!=\!1, \!k\!\in\!\{t,r\},\\
		&\label{P3_C12}\quad \quad p_k>0, \forall k \in\{t,r\}.
	\end{align}
\end{subequations}
Now, the main obstacle to solving this problem \eqref{ES_SDP} falls on the high coupling of user transmit power $\{p_k\}$ and uplink active/passive beamforming design $\{\mathbf{W}^{\textup{ES}}_k\}$ and $\{\widetilde{\mathbf{Q}}^{\textup{ES}}_k\}$ in constraints \eqref{P3_C1} and \eqref{P3_C2}. To overcome this issue, we first employ the BCD method to divide all variables into four blocks, i.e., $\{\mathbf{W}^{\textup{ES}}_k\}$, $\{p_k,\mathbf{V^{\textup{ES}}}\}$, $\{p_k, \widetilde{\mathbf{U}}_k^{\textup{ES}}\}$ and $\{\widetilde{\mathbf{Q}}^{\textup{ES}}_k\}$, then solve their corresponding subproblems iteratively. The following are the particulars:

\emph{1) Optimization for $\{\mathbf{W}^{\textup{ES}}_k\}$ With Fixed Other Variables:} As mentioned in \cite{Wang_MMSE}, the linear MMSE receiver is the optimal criterion for uplink NOMA. Inspired by this, for any given $\{\widetilde{\mathbf{Q}}^{\textup{ES}}_k\}$, we can determine the optimal solutions of $\mathbf{w}_k^{\textup{ES}}$ with constraint \eqref{P1_C2} as follows:
\begin{subequations}\label{ES_Receive}
	\begin{gather}
		\mathbf{w}^{\textup{ES}*}_t=\frac{\left(\widetilde{\mathbf{h}}_t\widetilde{\mathbf{h}}_t^H+\frac{\sigma^2}{p_t}\mathbf{I}_N\right)^{-1}\widetilde{\mathbf{h}}_t}{\bigg|\bigg|\left(\widetilde{\mathbf{h}}_t\widetilde{\mathbf{h}}_t^H+\frac{\sigma^2}{p_t}\mathbf{I}_N\right)^{-1}\widetilde{\mathbf{h}}_t\bigg|\bigg|},\\
		\mathbf{w}^{\textup{ES}*}_r=\frac{\left(\widetilde{\mathbf{h}}_r\widetilde{\mathbf{h}}_r^H+\frac{p_t}{p_r}\widetilde{\mathbf{h}}_t\widetilde{\mathbf{h}}_t^H+\frac{\sigma^2}{p_r}\mathbf{I}_N\right)^{-1}\widetilde{\mathbf{h}}_r}{\bigg|\bigg|\left(\widetilde{\mathbf{h}}_r\widetilde{\mathbf{h}}_r^H+\frac{p_t}{p_r}\widetilde{\mathbf{h}}_t\widetilde{\mathbf{h}}_t^H+\frac{\sigma^2}{p_r}\mathbf{I}_N\right)^{-1}\widetilde{\mathbf{h}}_r\bigg|\bigg|},
	\end{gather}
\end{subequations}
where $\widetilde{\mathbf{h}}_k=\mathbf{H}_k\widetilde{\mathbf{q}}^{\textup{ES}}_k \in \mathbb{C}^{N\times 1}, k\in\{t,r\}.$ Relying on this, we can further obtain the close-form solution of the $\mathbf{W}_k$ as $\mathbf{W}^{\textup{ES}*}_k=\mathbf{w}^{\textup{ES}*}_k(\mathbf{w}^{\textup{ES}*}_k)^H$.

\emph{2) Optimization for $\{p_k,\mathbf{V^{\textup{ES}}}\}$ with Fixed Other Variables}: For any given $\{\mathbf{W}^{\textup{ES}}_k,\widetilde{\mathbf{U}}_k^{\textup{ES}},\widetilde{\mathbf{Q}}^{\textup{ES}}_k\}$, problem \eqref{ES_SDP} is reduced to 
\begin{subequations}\label{ES_transmit1}
	\begin{align}
		&\quad \quad  \max_{{\{\mathbf{V}^{\textup{ES}}\},\{p_k\}}} \Gamma \\
		&\ {\rm s.t.} \ \eqref{P3_C1},\eqref{P3_C2},\eqref{P3_C3},\eqref{P3_C4},\eqref{P3_C5}, \eqref{P3_C12}.
	\end{align}
\end{subequations}
However, problem \eqref{ES_transmit1} is still intractable due to the non-convexity of constraints \eqref{P3_C2} and \eqref{P3_C5}. To this end, we further employ the successive convex approximation (SCA) technique to tackle constraint \eqref{P3_C2}. Specially, according to the first-order Taylor expansion of $R^{\textup{ES}}_r$ with respect to $p_t$, we can derive the lower boundary of $R^{\textup{ES}}_r$ as 
\begin{align}
	R^{\textup{ES}}_r\geq(1-\overline{\tau})\left(\log_2(X)-\log_2(Y)-R\right) \triangleq \widehat{R}_r^{\textup{ES}(l)},
\end{align}
where $X\!=\!p_r\mathrm{Tr}\!\left(\!\mathbf{W}^{\textup{ES}}_r\mathbf{H}_r\widetilde{\mathbf{Q}}^{\textup{ES}}_r\mathbf{H}^H_r\! \right)+p_t\mathrm{Tr}\!\left(\!\mathbf{W}^{\textup{ES}}_r\mathbf{H}_t\widetilde{\mathbf{Q}}^{\textup{ES}}_t\mathbf{H}^H_t\!\right)\!+\!\sigma^2$, $Y=p_t\mathrm{Tr}\left(\mathbf{W}^{\textup{ES}}_r\mathbf{H}_t\widetilde{\mathbf{Q}}^{\textup{ES}}_t\mathbf{H}^H_t\right)\!+\!\sigma^2 $, and $R\!=\!\frac{\mathrm{Tr}\left(\mathbf{W}^{\textup{ES}}_r\mathbf{H}_t\widetilde{\mathbf{Q}}^{\textup{ES}}_t\mathbf{H}^H_t\right)}{\textup{ln}2\left(p_t\mathrm{Tr}\left(\mathbf{W}^{\textup{ES}}_r\mathbf{H}_t\widetilde{\mathbf{Q}}^{\textup{ES}}_t\mathbf{H}^H_t\right)\!+\!\sigma^2\right)}(p_t\!-\!p^{(l)}_t)$. $l$ denotes the number of iterations. As for the non-convex rank-one constraint \eqref{P3_C5}, the SDR method is adopted to relax the original problem. Through the above conversion, the problem \eqref{ES_transmit1} is approximated as 
\begin{subequations}\label{ES_Transmit}
	\begin{align}
		& \max_{\{{\mathbf{V}^{\textup{ES}}\},\{p_k\}}} \Gamma \\
		& \ {\rm s.t.} \ \widehat{R}_r^{\textup{ES}(l)} \geq \Gamma, \\
		&\quad \quad \mathbf{V}^{\textup{ES}} \succeq 0, \forall k \in\{r,t\}, \\
		& \quad \quad \eqref{P3_C1}, \eqref{P3_C3}, \eqref{P3_C4}, \eqref{P3_C12}.
	\end{align}
\end{subequations}
It can be observed that the reformulated problem \eqref{ES_Transmit} is a standard SDP, which can be efficiently solved via existing convex solvers such as CVX \cite{cvx}. Actually, the following theorem reveals that the relaxed problem \eqref{ES_Transmit} has the same optimal solution as the original problem. Therefore, we can obtain the solutions to the original problem by directly solving the relaxed SDP problem.
\begin{theorem}
	For any given $P_A>0$, the optimal solutions for the simplified problem \eqref{ES_Transmit} always satisfy $\mathrm{Rank}(\mathbf{V^{\textup{ES}*}})=1$.
\end{theorem}
\begin{proof}
	Please refer to the Appendix.
\end{proof}
\emph{3) Optimization for $\{p_k, \widetilde{\mathbf{U}}_k^{\textup{ES}}\}$ With Fixed Other Variables}: 
In this subproblem, we jointly design $\{\widetilde{\mathbf{U}}^{\textup{ES}}_s\}$ and $\{p_k\}$ with fixed $\{\mathbf{W}^{\textup{ES}}_k,\mathbf{V}^{\textup{ES}},\widetilde{\mathbf{Q}}^{\textup{ES}}_k\}$. Accordingly, the original \eqref{ES_SDP} is transformed into a simplified form as follows: 
\begin{subequations}\label{ES_passive1}
	\begin{align}
		&\quad \max_{\{\widetilde{\mathbf{U}}^{\textup{ES}}_s\},\{p_k\}} \Gamma\\ 
		&\ {\rm s.t.}\eqref{P3_C1},\eqref{P3_C2},\eqref{P3_C4},\eqref{P3_C7},\eqref{P3_C9},\eqref{P3_C11}, \eqref{P3_C12}.
	\end{align}
\end{subequations}
It can be found that the main obstacles to the solution of the problem \eqref{ES_passive1} lie in the non-convex constraints \eqref{P3_C2} and \eqref{P3_C7}, where the former can be handled using a SCA technique similar to the previous subproblem. Thus, we focus primarily our efforts on the latter. To handle this rank-one constraint, we first rewrite it into an equivalent but tractable form as follows \cite{Yu_RANK}:
\begin{align}\label{Rank_one}
	\mathrm{Rank}(\widetilde{\mathbf{U}}^{\textup{ES}}_k)\!=1\! \Leftrightarrow \|\widetilde{\mathbf{U}}^{\textup{ES}}_k\|_{*}\!-\!\|\widetilde{\mathbf{U}}^{\textup{ES}}_k\|_{2}=0, \forall k \in\{t,r\},
\end{align}
where $\|\widetilde{\mathbf{U}}^{\textup{ES}}_k\|_{*}=\sum_i \sigma_i(\widetilde{\mathbf{U}}^{\textup{ES}}_k)$ and  $\|\widetilde{\mathbf{U}}^{\textup{ES}}_k\|_{2}=\sigma_1(\widetilde{\mathbf{U}}^{\textup{ES}}_k)$, and $\sigma_i$ is the $i$-th largest singular value of $\widetilde{\mathbf{U}}^{\textup{ES}}_k$. Next, by employing the penalty method \cite{penalty}, we add the right hand side of \eqref{Rank_one} as a penalty term to the original objective function as follows:
\begin{subequations}\label{ES_passive2}
	\begin{align}
		&\  \max_{\{\widetilde{\mathbf{U}}^{\textup{ES}}_s\},\{p_k\}}\; \Gamma-\xi\!\!\sum_{k\in\{t,r\}}\left(\|\widetilde{\mathbf{U}}^{\textup{ES}}_k\|_{*}-\|\widetilde{\mathbf{U}}^{\textup{ES}}_k\|_{2}\right)\\ 
		&\quad \quad {\rm s.t.} \  \ \widehat{R}_r^{\textup{ES}(l)} \geq \Gamma, \\
		& \quad \quad \quad \quad \widetilde{\mathbf{U}}^{\textup{ES}}_k \succeq 0, \forall k \in\{r,t\}, \\
		&\quad \quad \quad \quad \eqref{P3_C1},\eqref{P3_C4},\eqref{P3_C9},\eqref{P3_C11},\eqref{P3_C12},
	\end{align}
\end{subequations}
where $\xi>0$ denotes the penalty factor. To ensure the validity of the original optimization goal, $\xi$ should be initialized with a small value and then gradually increased until the following convergence criterion is met:
\begin{align}
	\max\{\|\widetilde{\mathbf{U}}^{\textup{ES}}_k\|_{*}-\|\widetilde{\mathbf{U}}^{\textup{ES}}_k\|_{2}, k\in\{t,r\}\} \leq \epsilon_1,
\end{align} 
where $\epsilon_1$ is a predefined threshold for constraining the iteration. However, this reformulated problem is still difficult to solve directly due to the non-convexity of the penalty term. To this end, we further apply the SCA technique to tackle this issue. With the aid of the first-order Taylor expansion, the non-convex penalty term can be approximated by its convex upper boundary as \eqref{Rank_one_SCA}, shown at the top of this page, 
\begin{figure*}[!t]
	\normalsize
	\begin{equation}\label{Rank_one_SCA}
		\begin{aligned}
			\|\widetilde{\mathbf{U}}^{\textup{ES}}_k\|_{*}\!\!-\!\|\widetilde{\mathbf{U}}^{\textup{ES}}_k\|_{2}
			&\leq \|\widetilde{\mathbf{U}}^{\textup{ES}}_k\|_{*}\!\!-\!\left\{\!\|\widetilde{\mathbf{U}}^{\textup{ES}(l)}_k\|_2\!+\!\mathrm{Tr}\!\left[\varphi \left(\widetilde{\mathbf{U}}^{\textup{ES}(l)}_k\right)\varphi \left(\widetilde{\mathbf{U}}^{\textup{ES}(l)}_k\right)^{H}\!\!\left(\widetilde{\mathbf{U}}^{\textup{ES}}_k\!\!-\! \widetilde{\mathbf{U}}_k^{\textup{ES}(l)}\right)\right]\!\right\} \\
			&\triangleq \|\widetilde{\mathbf{U}}^{\textup{ES}}_k\|_{*}-\overline{\mathbf{U}}^{\textup{ES}(l)}_k,
		\end{aligned}
	\end{equation}
	\hrulefill \vspace*{0pt}
\end{figure*}
where $\varphi\left(\widetilde{\mathbf{U}}^{\textup{ES}(l)}_s\right)$ denotes the eigenvector related to the largest eigenvalue of $\widetilde{\mathbf{U}}^{\textup{ES}(l)}_s$. Driven by this, problem \eqref{ES_passive2} is approximated to the following optimization problem:
\begin{align}\label{ES_passive_downlink}
	&\  \max_{\{\widetilde{\mathbf{U}}^{\textup{ES}}_s\},\{p_k\}}\; \Gamma-\xi\!\!\sum_{k\in\{t,r\}}\left(\|\widetilde{\mathbf{U}}^{\textup{ES}}_k\|_{*}-\overline{\mathbf{U}}^{\textup{ES}(l)}_k\right)\\ 
	&\quad {\rm s.t.} \quad \eqref{P3_C1},\eqref{P3_C4},\eqref{P3_C9},\eqref{P3_C11},\eqref{P3_C12}.
\end{align}
It can also be efficiently solved by CVX.

\emph{4) Optimization for $\{\widetilde{\mathbf{Q}}^{\textup{ES}}_k\}$ with Fixed Other Variables}: For any given $\{\mathbf{W}^{\textup{ES}}_k,\!\mathbf{V}^{\textup{ES}},\!\widetilde{\mathbf{U}}^{\textup{ES}}_k,p_k\}$, original problem is reformulated as 
\begin{subequations}\label{ES_passive3}
	\begin{align}
		&\quad \max_{\{\widetilde{\mathbf{Q}}^{\textup{ES}}_s\}}\ \Gamma\\ 
		&\ {\rm s.t.} \ \eqref{P3_C1},\eqref{P3_C2},\eqref{P3_C8},\eqref{P3_C10},\eqref{P3_C11}.
	\end{align}
\end{subequations}
As can be observed, the differential form of \eqref{P3_C2} with respect to $\widetilde{\mathbf{Q}}^{\textup{ES}}_t$ and $\widetilde{\mathbf{Q}}^{\textup{ES}}_r$ renders the problem still unsolvable. In this instance, we first introduce two auxiliary variables $A$ and $B$ with $\frac{1}{A}\leq\mathrm{Tr}\left(\mathbf{W}^{\textup{ES}}_r\mathbf{H}_r\widetilde{\mathbf{Q}}^{\textup{ES}}_r\mathbf{H}^H_r\right)$ and $B\geq \mathrm{Tr}\left(\mathbf{W}^{\textup{ES}}_r\mathbf{H}_t\widetilde{\mathbf{Q}}^{\textup{ES}}_t\mathbf{H}^H_t\right)+\sigma^2$, respectively. Therefore, constraint \eqref{P3_C2} can be rewritten as 
\begin{align}
	\left(1-\overline{\tau}\right)\log_2\left(1+\frac{1}{AB}\right)\geq \Gamma,
\end{align}
Next, using the first Taylor expansion of $\log_2\left(1+\frac{1}{AB}\right)$ with respect to $A$ and $B$, we may obtain its convex lower boundary as follows:
\begin{align}
	\log_2(1+\frac{1}{AB}) &\!\geq \log_2(1+\frac{1}{A^{(l)}B^{(l)}})-\frac{\log_2(e)\left(A-A^{(l)}\right)}{A^{(l)}\left(1+A^{(l)}B^{(l)}\right)} \nonumber\\
	&\!- \frac{\log_2(e)\left(B-B^{(l)}\right)}{B^{(l)}\left(1+A^{(l)}B^{(l)}\right)} \triangleq \widetilde{R}^{\textup{ES}(l)}_r/(1-\overline{\tau}),
\end{align}
where $A^{(l)}$ and $B^{(l)}$ denote the local points of $A$ and $B$ in the $l$-th iteration, respectively. As for the non-convex rank-one constraint in \eqref{P3_C8}, similar penalty method can be utilized. Consequently, problem \eqref{ES_passive3} is modified to:
\begin{subequations}\label{ES_passive_uplink}
	\begin{align}
		&\quad \max_{A,B,\{\widetilde{\mathbf{Q}}^{\textup{ES}}_k\}}  \ \Gamma-\xi_1\!\!\sum_{k\in\{t,r\}}\left(\|\widetilde{\mathbf{Q}}^{\textup{ES}}_k\|_{*}-\overline{\mathbf{Q}}^{\textup{ES}(l)}_k\right) \\
		&\quad \quad {\rm s.t.} \ \ \widetilde{R}_r^{\textup{ES}(l)}\geq \Gamma,\\
		&\quad \quad \quad \quad \frac{1}{A}\leq\mathrm{Tr}\left(\mathbf{W}^{\textup{ES}}_r\mathbf{H}_r\widetilde{\mathbf{Q}}^{\textup{ES}}_r\mathbf{H}^H_r\right),\\
		&\quad \quad\quad \quad B \geq \mathrm{Tr}\left(\mathbf{W}^{\textup{ES}}_r\mathbf{H}_t\widetilde{\mathbf{Q}}^{\textup{ES}}_t\mathbf{H}^H_t\right)+\sigma^2,\\
		&\quad \quad\quad \quad \widetilde{\mathbf{Q}}^{\textup{ES}}_k \succeq 0, \forall k \in \{t,r\},\\
		&\quad \quad\quad \quad \eqref{P3_C1},\eqref{P3_C10},\eqref{P3_C11}.
	\end{align}
\end{subequations}
This problem is also a standard SDP, so it can be solved directly by CVX.

Based on the above analysis, the specifics of the developed two-layer algorithm for solving the MTM problem associated with the ES-NOMA strategy are provided in \textbf{Algorithm 1}. To be more specific, on the one hand, the objective value of the original problem \eqref{ES_SDP} with an upper bound is updated as problems \eqref{ES_Receive}, \eqref{ES_Transmit}, \eqref{ES_passive_downlink}, and \eqref{ES_passive_uplink} are solved alternately, which is non-decreasing. On the other hand, the optimal time allocation is always determined after several one-dimensional searches. Therefore, the convergence of the two-layer iterative \textbf{Algorithm 1} is guaranteed.

\begin{remark}
	\textup{The \textbf{Algorithm 1} can also be extended to the K-user scenario, where $K>2$. In the downlink WPT phase, the HAP propagates energy in a broadcast form, which is unaffected by the number of users. Therefore, the design of the downlink joint beamforming can still be completed by directly solving problems (21) and (27). For the uplink WIT phase, we may first sort the channel conditions for all users and then derive the close-form solution of HAP receive beamforming for each user according to the MMSE criterion. Subsequently, we update the $R^{\textup{ES}}_k$ expression for each user according to the uplink NOMA decoding order and incorporate it into (32) to solve the problem.}
\end{remark}
\begin{algorithm}[!t]\label{method1}
	\caption{Proposed two-layer iterative algorithm to solve problem \eqref{ES_CTM}.}
	\label{alg:A}
	\begin{algorithmic}[1]
		\STATE {Initialize the variables $\{\widetilde{\mathbf{Q}}_k^{\textup{ES(0)}},\widetilde{\mathbf{U}}_k^{\textup{ES(0)}}, p^{(0)}_k\}$, $\tau_0=0$, $ \varepsilon_0$, and the search step $\triangle\tau$.}
		\REPEAT  
		\STATE {Set the time allocation $\tau_0=\tau_0+\triangle\tau$ and the iteration number $l=1$} 
		\REPEAT 
		\STATE { For given $\{\widetilde{\mathbf{Q}}^{\textup{ES}(l-1)}\}$, determine the $\{\mathbf{W}^{\textup{ES}(l)}_k\}$ according to  (\ref{ES_Receive})}.
		\STATE {For given $\{\mathbf{W}^{\textup{ES}(l)},\widetilde{\mathbf{Q}}^{\textup{ES}(l-1)},\widetilde{\mathbf{U}}^{\textup{ES}(l-1)}\}$, solve problem (\ref{ES_Transmit}), update $\{\mathbf{V}_k^{\textup{ES}(l)}\}$ and $\{p^{(l)}_k\}$}.
		\STATE {For given $\{\mathbf{W}^{\textup{ES}(l)},\widetilde{\mathbf{Q}}^{\textup{ES}(l-1)},{\mathbf{V}}_k^{\textup{ES}(l)}\}$, solve problem (\ref{ES_passive_downlink}), update $\{\widetilde{\mathbf{U}}^{\textup{ES}(l)}\}$ and $\{p^{(l)}_k\}$}.\\
		\STATE {For given $\{\mathbf{W}^{\textup{ES}(l)},{\mathbf{V}}_k^{\textup{ES}(l)},\widetilde{\mathbf{U}}^{\textup{ES}(l)},p^{(l)}_k\}$, solve problem (\ref{ES_passive_uplink}), update $\{\widetilde{\mathbf{Q}}^{\textup{ES}(l)}\}$}.\\
		\STATE Update $l \leftarrow l+1$.\\
		\UNTIL the iteration yield is below $\varepsilon_0$.\\
		\UNTIL $\tau_0>1$.
		\STATE {Output} the optimal solutions
		$\{\tau^*_0\}$= $\textup{argmax}_{\tau_0}\Gamma$, $\mathbf{V}^{\textup{ES}*}=\mathbf{V}^{\textup{ES}(l)}$, $\mathbf{W}^{\textup{ES}*}_k=\mathbf{W}^{\textup{ES}(l)}_k$, $\!\{p^{*}_k\}\!=\!\{p^{l}_k\}$, $\{\widetilde{\mathbf{U}}_k^{\textup{ES}*}\}=\{\widetilde{\mathbf{U}}_k^{\textup{ES}(l)}\}$ and $\{\widetilde{\mathbf{Q}}_k^{\textup{ES}*}\}=\{\widetilde{\mathbf{Q}}_k^{\textup{ES}(l)}\}$ under $\tau^{*}_0$.
	\end{algorithmic}
\end{algorithm}

\subsection{Proposed Solution for TS-TDMA}
It should be noticed that due to the time switching feature, only one user is served by the STAR-RIS, for any given transmission/reflection-beamforming vector. Thus, the optimal active beamforming vectors at the HAP are the MRT beamformers, i.e., $\mathbf{v}^{\textup{TS}*}_k=\sqrt{P_k}\frac{\mathbf{G}^H_k\widetilde{\mathbf{u}}^{\textup{TS}}_s}{\|\mathbf{G}^H_k\widetilde{\mathbf{u}}^{\textup{TS}}_s\|_2}$ and $\mathbf{w}^{\textup{TS}*}_k=\frac{\mathbf{H}^H_k\widetilde{\mathbf{q}}^{\textup{TS}}_s}{\|\mathbf{H}^H_k\widetilde{\mathbf{q}}^{\textup{TS}}_s\|_2}$ for downlink transmit beamforming and uplink receive beamforming, respectively, where $P_k$ denotes the HAP transmit power allocated to user $k$. Besides, since no inter-user interference exists in the TS-TDMA transmission strategy, both users will use up their harvested energy. Therefore, by substituting $\{\mathbf{v}^{\textup{TS}*}_k\}$ and $\{\mathbf{w}^{\textup{TS}*}_k\}$ into problem \eqref{TS_CTM}, the optimal passive beamforming for both downlink $\{\widetilde{\mathbf{u}}^{\textup{TS}}_k\}$ and uplink $\{\widetilde{\mathbf{q}}^{\textup{TS}}_k\}$ are given by
\begin{subequations}\label{TS_passive}
	\begin{align}
		& \ \ \widetilde{\mathbf{u}}^{\textup{TS}*}_k= \textup{argmax} \ (\widetilde{\mathbf{u}}^{\textup{TS}}_k)^H\mathbf{G}_k\mathbf{G}^H_k\widetilde{\mathbf{u}}^{\textup{TS}}_k, \forall k \in \{t,r\}, \\
		&\ \ \widetilde{\mathbf{q}}^{\textup{TS}*}_k= \textup{argmax} \ (\widetilde{\mathbf{q}}^{\textup{TS}}_k)^H\mathbf{H}^H_k\mathbf{H}_k\widetilde{\mathbf{q}}^{\textup{TS}}_k, \forall k \in \{t,r\},\\
		&{\rm s.t.} \ | [\widetilde{\mathbf{u}}^{\textup{TS}}_k]_m \| =1,  \ | [\widetilde{\mathbf{q}}^{\textup{TS}}_k]_m|=1, m=1,\cdots,M+1.
	\end{align}
\end{subequations}
After obtaining the optimal active/passive beamforming for both downlink WPT and uplink WIT, problem is reduced to a resource allocation problem as 
\begin{subequations}\label{TS_resource}
	\begin{align}
		&\  \max_{{\{P_k\},\{\tau_{0,k}\},\{\tau_{1,k}\}}} \Gamma \\
		&\label{P4_C1}{\rm s.t.}\ \ \tau_{0,k}>0,\tau_{1,k}>0,P_k>0, \forall k\in\{t,r\},\\
		&\label{P4_C2}\quad \quad \tau_{0,t}+\tau_{0,r}+\tau_{1,t}+\tau_{1,r}=T,\\ 
		&\label{P4_C3}\quad \quad P_k\leq P_A,\\ 
		&\label{P4_C4}\quad \quad \tau_{1,k}\log_2\left(1+\frac{\tau_{0,k}P_k\|\mathbf{G}_k\mathbf{u}^{\textup{TS}*}_k\|^2*\|\mathbf{H}_k\mathbf{q}^{\textup{TS}*}_k\|^2}{\tau_{1,k}\sigma^2}\right) \geq \Gamma, \nonumber\\
		&\quad \quad k\in\{t,r\}. 
	\end{align}
\end{subequations}
As can be seen, the equation in constraint \eqref{P4_C3} always holds when the problem \eqref{TS_resource} is optimal; otherwise, we can continue to increase $P_k$ while keeping other variables constant and obtain a better solution. Furthermore, \eqref{P4_C1} and \eqref{P4_C2} are linear constraints on the time allocation variables, and constraint \eqref{P4_C4} is a jointly convex form for $\tau_{0,k}$ and $\tau_{1,k}, k\in\{t,r\}$. Thus, problem \eqref{TS_resource} is a standard convex problem and can be solved directly by CVX. At this point, the proposed algorithm for the TS-TDMA strategy can be summarized as \textbf{Algorithm 2}. Apparently, when $P_A>0$, problems \eqref{TS_passive} and \eqref{TS_resource} are always feasible and restricted by a finite boundary value. Hence, the proposed \textbf{Algorithm 2} is feasible and convergent.
\begin{remark}
	\textup{The \textbf{Algorithm 2} can be directly extended to the system with $K>2$ users. First, solve problem (33) to obtain the optimal beamforming design for each user in both the WPT and WIT phases. Then, solve problem (34) with more complex time allocation constraints to obtain the optimal resource allocation strategy.}
\end{remark}
\begin{algorithm}[!t]\label{method2}
	\caption{Proposed algorithm to solve \eqref{TS_CTM}.}
	\label{alg:B}
	\begin{algorithmic}[1]
		\STATE {Determine optimal active transmission and reception vectors $\mathbf{v}_k^{\textup{TS}*}$ and $\mathbf{w}_k^{\textup{TS}*}$ with MRT beamformer.} \\
		\STATE {Solve problem \eqref{TS_passive} to obtain the optimal passive transmission/reflection-coefficient vectors $\widetilde{\mathbf{u}}_k^{\textup{TS}*}$ and $\widetilde{\mathbf{q}}_k^{\textup{TS}*}$ for downlink WPT and uplink WIT, respectively}.\\
		\STATE{Solve problem \eqref{TS_resource} with given $\{\widetilde{\mathbf{q}}_k^{\textup{TS}},\widetilde{\mathbf{u}}_k^{\textup{TS}} , k\in\{t,r\}\}$ to obtain the optimal time allocation $\tau^{*}_{0,k}$ and $\tau^{*}_{1,k}$}.
		\STATE {Output} the optimal solutions $\Gamma^*$ with
		$\{\tau^*_{0,k},\tau^*_{1,k}\}$, $\{\mathbf{v}_k^{\textup{TS}*}\}$, $\{\mathbf{w}_k^{\textup{TS}*}\}$, $\{\widetilde{\mathbf{u}}_k^{\textup{TS}*}\}$ and $\{\widetilde{\mathbf{q}}_k^{\textup{TS}*}\}$.
	\end{algorithmic}
\end{algorithm}

\subsection{Computational Complexity Analysis}
On the one hand, for ES-NOMA, with fixed time allocations $\tau_0$ and $\tau_1$ in the outer layer of \textbf{Algorithm 1}, the original problem \eqref{ES_SDP} can be transformed into four simpler subproblems by employing the BCD strategy. Of these, the optimal receive beamforming is determined directly based on MMSE guidelines, while the remaining three subproblems are converted to the standard SDP, which can be solved by the interior point method \cite{SDR}. Thus, the approximate computational complexity of \textbf{Algorithm 1}  is given by $\mathcal{O}^{\textup{ES}}=\mathcal{O}\left(I_\textup{out}I_\textup{in}\left(KN^{3.5}+4M^{3.5}\right)\right)$, where $I_\textup{out}$ and $I_\textup{in}$ denote the number of outer and inner loops, respectively, and $K$ indicates the number of users. On the other hand, for TS-TDMA, the optimal passive beamforming for both downlink and uplink can be determined by employing the suboptimal SDR based method. Then, the computational complexity is $\mathcal{O}\left(4M^{3.5}\right)$. Besides, the computational complexity of the convex problem \eqref{TS_resource} is $\mathcal{O}\left(K^{3.5}\right)$. Thus, the approximate computational complexity of \textbf{Algorithm 2} is given by $\mathcal{O}^{\textup{TS}}=\mathcal{O}\left(4M^{3.5}+K^{3.5}\right)$. It can be noticed that the lower computational complexity of TS-TDMA compared to ES-NOMA can be attributed to the fewer optimization variables and iterations required.

\section{Numerical Results}
In this section, we present numerical results from various perspectives to demonstrate the contribution of STAR-RIS to the fairness performance of WPCNs.
\subsection{Simulation Setup}
As illustrated in Fig. \ref{simulation}, a three-dimensional (3D) coordinate setup is taken into consideration, with the HAP located at $(0,0,2)$ meters and the STAR-RIS deployed at $(10,0,0)$ meters. Users $t$ and users $r$ are randomly distributed in the T and R semicircular regions centered on STAR-RIS with a radius of $r=1$ m, respectively. In addition, the Rician fading model is employed in this paper, where the path loss exponents for the HAP-STAR-RIS and the STAR-RIS-user $k$ links are set to be $\alpha_{a,s}=\alpha_{s,k}=2.2$, while those for the HAP-user $k$ links are set to be $\alpha_{a,k}=3.4$ \cite{Wu_dynamic}. Furthermore, the communication bandwidth and carrier frequency are assumed to be $B=1$ MHz and $f=750$ MHz, respectively. Other parameters are set to be $P_A=5$ W, $\eta=0.8$, $\sigma^2 = -90$ dBm, $\varepsilon_0=\varepsilon_1=10^{-5}$, $\triangle\tau=0.1$s.
\begin{figure}[t]
	\setlength{\abovecaptionskip}{0cm}   
	\setlength{\belowcaptionskip}{0cm}   
	\setlength{\textfloatsep}{7pt}
	\centering
	\includegraphics[width=3.5in]{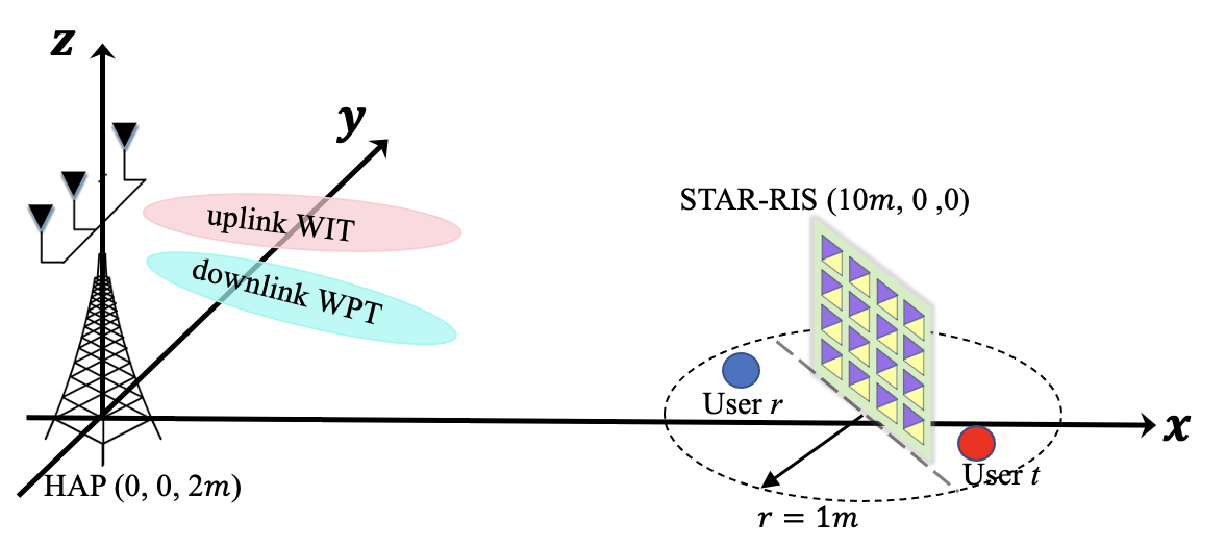}
	\caption{Simulation setup.}
	\label{simulation}
\end{figure}

In this paper, we include the following three baseline schemes for comparison: 1) \textbf{w/o RIS-NOMA/TDMA}: In this case, communication between the HAP and the users is guaranteed by the direct link without the help of RIS. 2) \textbf{RIS-NOMA/TDMA}: In this case, a conventional reflecting-only RIS is employed to facilitate the WPCN. In particular, to fulfill all potential user distribution coverage service requirements, the RIS is deployed at $(10,1,0)$ meters. Note that since no new variables are introduced, our proposed algorithms can be directly applied to the above baseline schemes with appropriate modifications. 3) \textbf{STAR-NOMA, GR}\cite{Wu_Gaussian}: In this case, the Gaussian randomization is applied to tackle the rank-one constraints for both STAR-RIS beamforming $\widetilde{\mathbf{U}}^{\textup{ES}}_k$ and $\widetilde{\mathbf{Q}}^{\textup{ES}}_k$ in the formulated ES-NOMA problem. Besides, the computational complexity of this method is given by $\mathcal{O}^{\textup{GR}}=\mathcal{O}\left(I_{\textup{int}}\left(KN^{3.5}+4M^{3.5}+4MD\right)\right)$, where $I_{\textup{int}}$ denotes the number of iterations required for convergence, and $D$ denotes the number of randomizations, which is typically set to 1000.
\subsection{Convergence of  Proposed Algorithm 1}
\begin{figure}[t]
	\setlength{\abovecaptionskip}{0cm}   
	\setlength{\belowcaptionskip}{0cm}   
	\setlength{\textfloatsep}{7pt}
	\centering
	\includegraphics[width=3.5in]{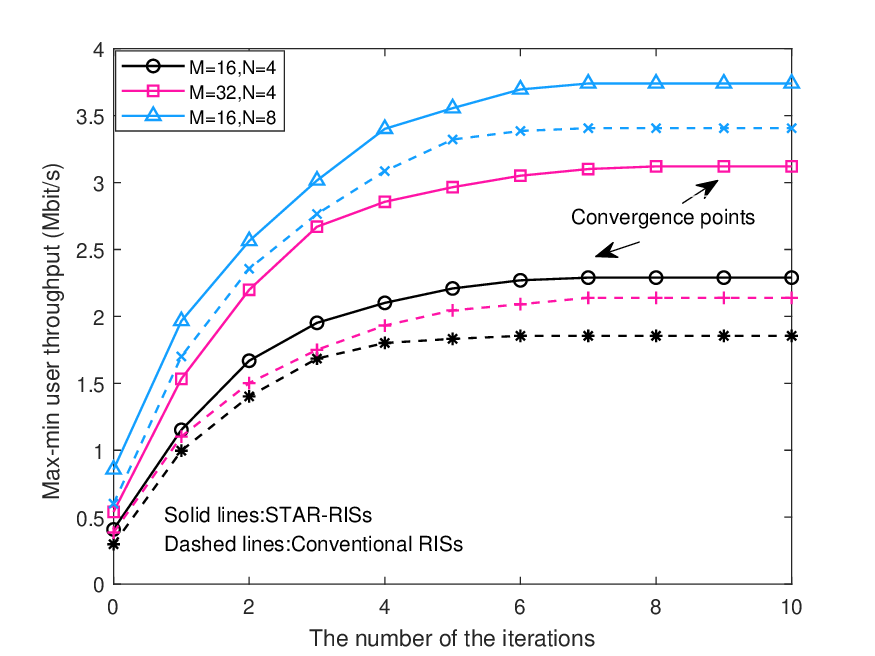}
	\caption{Convergence of Algorithm 1.}
	\label{Convergence}
\end{figure}
In Fig. \ref{Convergence}, we show the convergence behavior of the proposed \textbf{Algorithm 1} for STAR-RISs and conventional RISs under various settings of HAP and STAR-RIS. The results show that for all scenarios, the max-min user throughput rapidly rises with an increase in iterations and eventually stabilizes to a value within 10 iterations. However, with an increase in the number of HAP antennas or STAR-RIS elements, the convergence rate will gradually slow down due to an increase in variables that must be optimized. Furthermore, STAR-RISs require a greater number of iterations to attain convergence points for larger values of $M$ when compared to conventional RISs. This result is unsurprising since the extra DoFs facilitated by STAR-RISs augment the system's design complexity while generating improvements.

\subsection{Fairness Performance in a Single-antenna HAP Case}
\begin{figure}[t]
	\setlength{\abovecaptionskip}{0cm}   
	\setlength{\belowcaptionskip}{0cm}   
	\setlength{\textfloatsep}{7pt}
	\centering
	\includegraphics[width=3.5in]{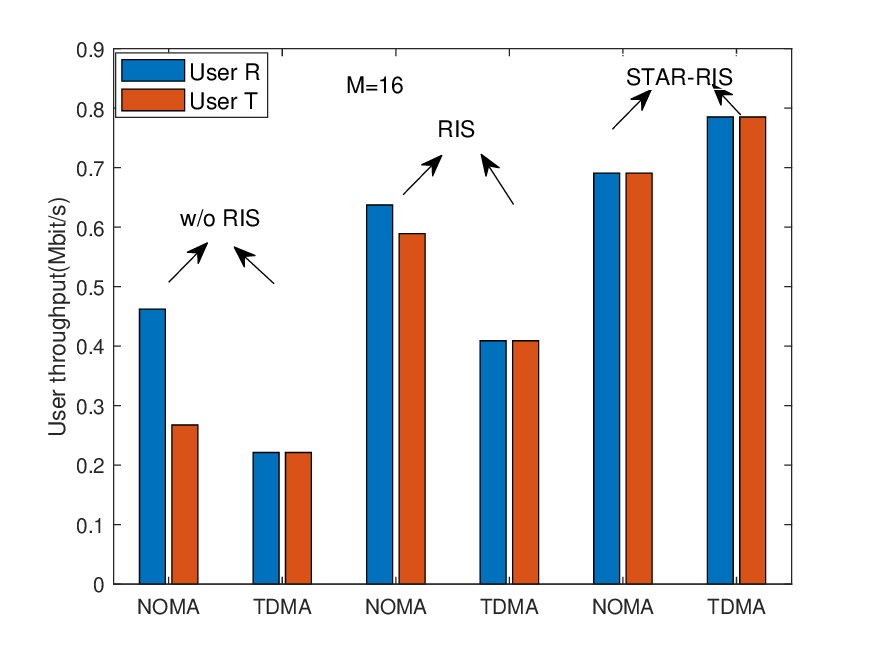}
	\caption{User throughput in a single-antenna HAP case.}
	\label{single-antenna}
\end{figure}
In Fig. \ref{single-antenna}, we study the system fairness performance in a particular single-antenna HAP case. Here, we consider $M$ = 16. The results depict that for baseline schemes, the TDMA strategy bridges the user performance gap better, but the NOMA strategy achieves a better minimum user throughput. This is because the flexible allocation of time resources in TDMA dilutes the impact of channel condition gaps on users, but also results in lower time utilization per user. However, in the case of STAR-RIS, TDMA is superior to NOMA in terms of achievable minimum user throughput. This can be interpreted to mean that the ES protocol enabled by STAR-RIS in the NOMA strategy utilizes a form of energy splitting to promote user fairness. In this case, the single-antenna HAP has difficulty addressing the energy leakage and interference management issues in NOMA transmission, thus curbing individual user performance.
\subsection{Max-min Throughput Versus Number of HAP Antennas}
\begin{figure}[t]
	\setlength{\abovecaptionskip}{0cm}   
	\setlength{\belowcaptionskip}{0cm}   
	\setlength{\textfloatsep}{7pt}
	\centering
	\includegraphics[width=3.5in]{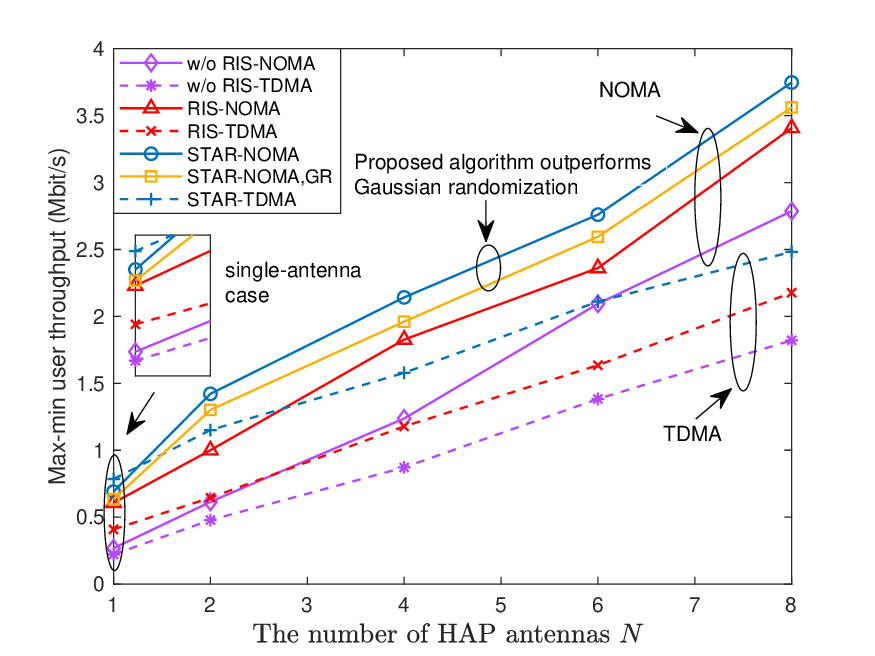}
	\caption{Max-min throughput versus number of HAP antennas.}
	\label{versus_antenna}
\end{figure}
In Fig. \ref{versus_antenna}, we investigate the max-min user throughput in relation to the number of HAP antennas. We set $M = 16$. In the single-antenna case, STAR-TS-TDMA outperforms the other schemes due to its interference-free communication and improved DoFs. In addition, the active beamforming gain at the HAP drives the increase in the achievable max-min user throughput for all schemes as $N$ increases. The difference is that NOMA strategies grow faster than TDMA strategies. The reason is that the DoFs from additional active antennas at the HAP are effective in canceling inter-user interference when receiving signals, which significantly compensates for the weaker interference handling capability of NOMA compared to TDMA. Besides, for NOMA strategies, both users have been recharging throughout the WPT phase. Therefore, as $N$ increases, each user can harvest more energy from the WPT phase to boost the data upload performance. With regard to the performance of STAR-RIS compared to the baseline schemes, STAR-RIS consistently outperforms in both NOMA and TDMA due to its natural deployment advantages and additional DoF gains. In addition, although the computational complexity of our proposed penalty-based algorithm is generally higher than that of the Gaussian randomization method due to its nested iterative structure, one can observe from Fig. \ref{versus_antenna} that our algorithm notably outperforms Gaussian randomization in terms of achievable performance. This is because the Gaussian randomization method first solves a relaxed problem without rank-one constraints using SDR, then reconstructs the solution via random sampling. Therefore, it can only offer an approximate solution to the original problem. In contrast, our algorithm directly incorporates the rank-one constraints into the optimization process as penalties applied over several iterations. Although this introduces additional complexity, it significantly improves the quality of the solution.

\subsection{Max-min Throughput Versus Number of STAR-RIS Elements}
\begin{figure}[t]
	\setlength{\abovecaptionskip}{0cm}   
	\setlength{\belowcaptionskip}{0cm}   
	\setlength{\textfloatsep}{7pt}
	\centering
	\includegraphics[width=3.5in]{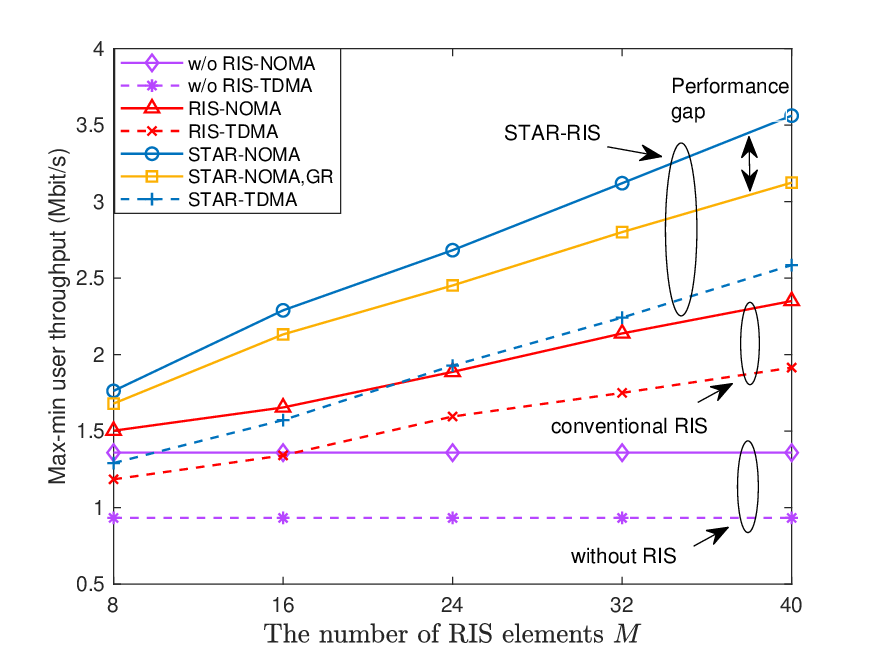}
	\caption{Max-min throughput versus number of STAR-RIS elements.}
	\label{versus_element}
\end{figure}
In Fig. \ref{versus_element}, we explore the achievable max-min user throughput versus the number of STAR-RIS elements $M$ with $N=4$. The results show that the max-min user throughput increases with $M$ for both the STAR-RIS and the conventional RIS because of the passive beamforming gain. However, the STAR-RIS rises noticeably faster than the conventional RIS. This is because the extra DoFs for the STAR-RIS allow the passive beamforming gains to be extended by a greater number of elements, which also explains why RIS-NOMA outperforms STAR-TS-TDMA when $M$ is small but vice versa when $M$ is large. Moreover, the performance gap between the different strategies for STAR-RIS widens as $M$ increases. This is anticipated as the energy leakage in ES can be effectively mitigated when $M$ is large, but the time utilization in TS is not improved. Besides, the performance gain of our proposed algorithm over the Gaussian randomization method increases as $M$ increases. This is because, as $M$ grows, the random sampling and normalization features of Gaussian randomization struggle to ensure that each element reaches its near-optimal point, limiting its performance gain. In contrast, our proposed algorithm optimizes all elements as a whole with higher computational effort, thus ensuring a high-quality solution regardless of the number of elements. This demonstrates that our proposed algorithm is a more effective and robust solution. Moreover, the trade-offs between solution quality and computational effort can provide useful guidelines for practical STAR-RIS implementation in WPCNs.

\subsection{Max-min Throughput Versus Time Allocation}
\begin{figure}[]
	\subfigure[For different numbers of antennas $N$.]{\label{versus_time1}
		\includegraphics[width= 3.5in]{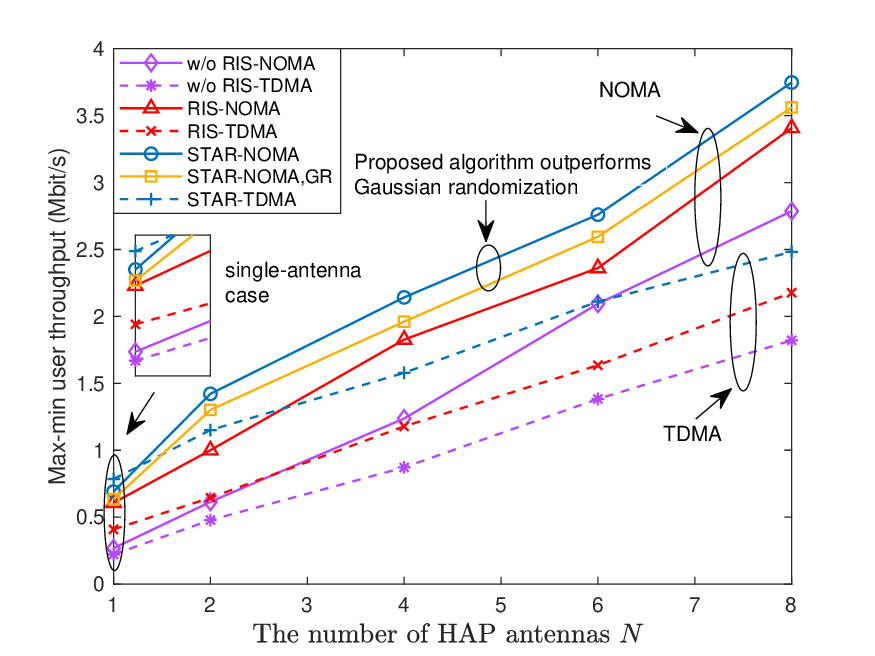}}
	\subfigure[For different numbers of elements $M$.]{\label{versus_time2}
		\includegraphics[width= 3.5in]{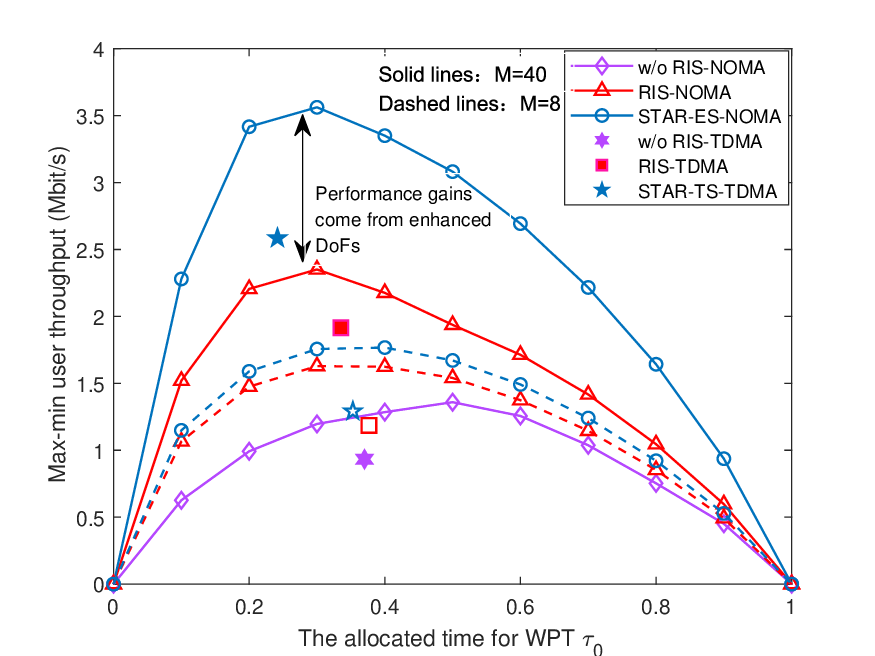}}
	\caption{\textcolor{black}{Max-min throughput versus time allocation.} } 
	\label{versus_time}
\end{figure}
In Fig. \ref{versus_time}, we examine the achievable max-min user throughput versus the time allocation in WPCN. Particularly, the results presented in Fig. \ref{versus_time1} are based on different numbers of antennas $N$ when $M=16$. It can be found for NOMA, the max-min user throughput first increases as $\tau_0$ increases until it reaches the optimal $\tau^{*}_0$, and then begins to decline as $\tau_0$ continues to grow. The reasons for this can be explained in the following way. When $\tau_0$ is small, i.e., the left region of $\tau^*_0$, users can only harvest a limited amount of energy from the HAP for the subsequent WIT phase, which severely weakens the system performance. Therefore, increasing the allocated time for WPT in this interval can provide users with more sufficient energy for information transfer, thus improving the max-min user throughput. However, as $\tau_0$ increases, the time available for WIT decreases correspondingly; when $\tau_0$ exceeds $\tau^*_0$, the energy recharge from the increase in time for WPT cannot compensate for the loss of time for WIT. Thus, the max-min user throughput also decreases as $\tau_0$ increases. In other words, $\tau^*_0$ strikes the best tradeoff between WPT and WIT. Besides, we also observe that the optimal time allocated to the WPT decreases as $N$ increases for both NOMA and TDMA. This makes sense because the active beamforming gains from additional antennas are able to enhance the transmission efficiency in both phases, thus allowing the users to spend less time harvesting sufficient energy and more time transmitting information. 

In Fig. \ref{versus_time2}, we extend the above study for different numbers of elements $M$ based on a special antenna case of $N=4$. Similar results to Fig. \ref{versus_time1} show that the max-min user throughput for NOMA schemes first increases and then decreases with $\tau_0$ increasing. In addition, the optimal time allocation for WPT $\tau^*_0$ decreases while the achievable max-min user rate increases for all schemes with the elements $M$ grows. To be more specific, the improved DoFs allow the STAR-RIS to outperform the conventional RIS by a large margin.
\subsection{Max-min Throughput Versus STAR-RIS Deployment Locations}
\begin{figure}[t]
	\setlength{\abovecaptionskip}{0cm}   
	\setlength{\belowcaptionskip}{0cm}   
	\setlength{\textfloatsep}{7pt}
	\centering
	\includegraphics[width=3.5in]{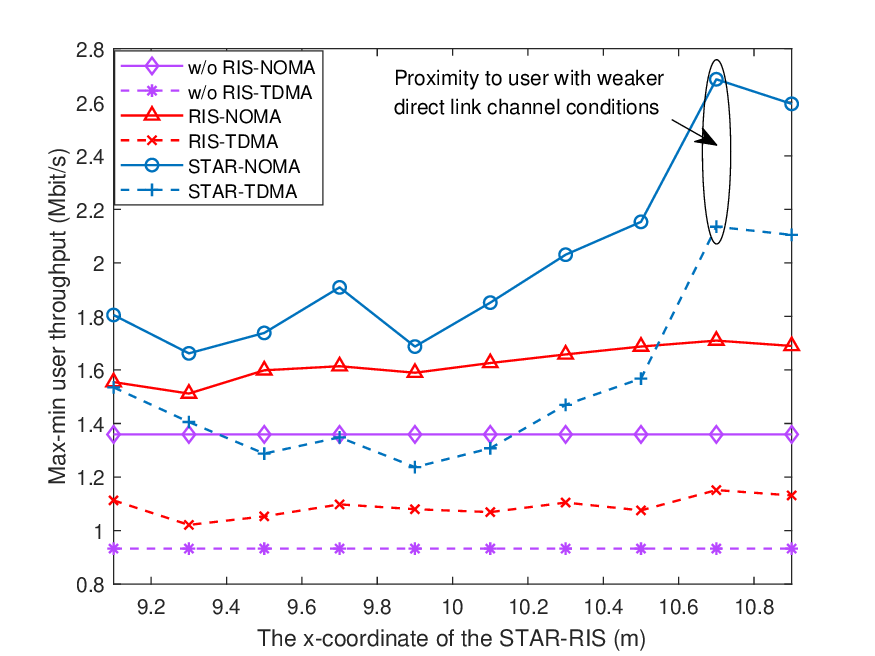}
	\caption{Max-min throughput versus STAR-RIS deployment locations.}
	\label{versus_location}
\end{figure}
In Fig. \ref{versus_location}, we study the effect of the STAR-RIS deployment locations on the achievable max-min user throughput. Herein, we consider a deterministic scenario for user $r$ located at $(9,0,0)$ meters and user $t$ located at $(11,0,0)$ meters, where the STAR-RIS is deployed on two user connection lines with a y-coordinate and a
z-coordinate of 0. In addition, we set $M=16$, $N=4$. As can be seen, deploying STAR-RIS/RIS close to the user $t$ can achieve better performance. This is because user $t$ suffers more severe direct path loss than user $r$ due to its longer communication distance with the HAP, which results in unfair performance between users in the traditional WPCN. Whereas deploying STAR-RIS near user $t$ can provide it with better indirect channel conditions than user $r$, which effectively reduces the gap in channel conditions between the two due to the direct links and makes the resource allocation more flexible. Therefore, a higher performance gain can be achieved. Besides, this phenomenon is more pronounced in STAR-RIS than in conventional RIS. This is made possible by the fact that the comprehensive coverage feature enables STAR-RIS to be deployed more reasonably and flexibly, so that the combined channel conditions from the HAP through the STAR-RIS to different users can be coordinated more efficiently.

\subsection{\textcolor{black}{Max-min Throughput Versus Number of Users}}
\begin{figure}[t]
	\setlength{\abovecaptionskip}{0cm}   
	\setlength{\belowcaptionskip}{0cm}   
	\setlength{\textfloatsep}{7pt}
	\centering
	\includegraphics[width=3.5in]{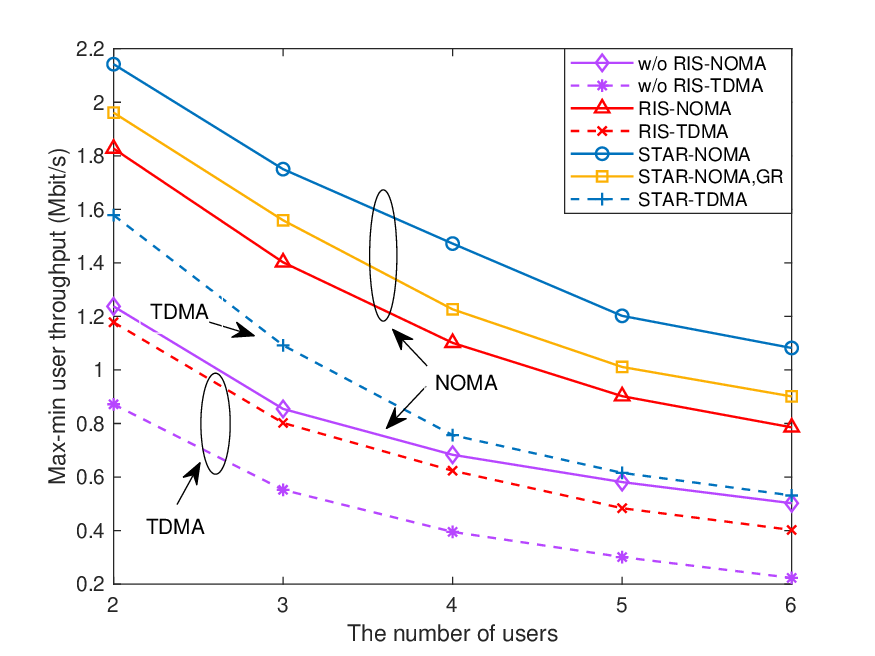}
	\caption{\textcolor{black}{Max-min throghput versus number of users.}}
	\label{versus_user}
\end{figure}
With appropriate modifications, our proposed algorithm can be extended to multi-user scenarios. To illustrate this, we study the max-min throughput versus the number of users in Fig. \ref{versus_user}. As can be observed, the max-min throughput decreases as the number of users increases for all strategies. This is expected, as new users inevitably take up some of the communication resources allocated to existing users to maintain fairness among all users. In particular, this phenomenon is more pronounced in TDMA strategies. This is because when the number of users increases, the efficiency of exploiting communication time for TDMA strategies is further degraded since it cannot allow all users to be served at the same time. In addition, in NOMA strategies, our proposed scheme for STAR-RISs consistently achieves the best performance because it has additional DoFs to eliminate inter-user interference in communication. 

\section{Conclusions}
In this paper, a STAR-RIS assisted WPCN was studied, where two transmission strategies were proposed, namely ES-NOMA and TS-TDMA. Based on this, the minimum user throughput maximization optimization problem was formulated for both strategies. To solve the resulting intractable problem, a two-layer based BCD framework was proposed for ES-NOMA and a simple method for the convex problem was utilized in TS-TDMA. Numerical results unveiled that STAR-RISs can significantly outperform conventional reflecting-only RISs for both transmission strategies. Furthermore, TS-TDMA is capable of achieving better suppression of the doubly-near-far effect in WPCNs for a single-antenna HAP case. However, in the case of multi-antenna HAP, ES-NOMA outperforms TS-TDMA, and the performance gap between the two becomes more pronounced as the number of STAR-RIS elements increases.

Given the growing research efforts in STAR-RIS prototyping \cite{STAR-RIS_implementation} and practical RIS-assisted WPT \cite{SWIPT_implementation}, these design insights obtained from the theoretical study in this work can provide useful guidance for the practical evaluation and operation of STAR-RIS-assisted WPCNs in the future work.

\section*{Appendix: Proof of Theorem 1} \label{Appendix:A}
Here, we start by recalling that the relaxed problem \eqref{ES_Transmit} has the following form:
\begin{subequations}\label{ES_Transmit1}
	\begin{align}
		& \max_{\{{\mathbf{V}^{\textup{ES}}\},\{p_k\}}} \Gamma \\
		& \ {\rm s.t.} \ R_t \geq \Gamma, \\
		&\quad \quad \widehat{R}_r^{\textup{ES}(l)} \geq \Gamma, \\
		&\quad \quad  \mathrm{Tr}\left(\mathbf{V}^{\textup{ES}}\right) \leq P_A, \\
		&\quad \quad p_k >0, \forall k \in \{t,r\},\\
		& \quad \quad (1-\overline{\tau})p_k\! \leq \! \overline{\tau}\mathrm{Tr}\left(\!\mathbf{V^{\textup{ES}}}\mathbf{G}^H_k\widetilde{\mathbf{U}}^{\textup{ES}}_k\mathbf{G}_k\!\right), k\!\in\!\{t,r\},\\
		&\quad \quad \mathbf{V}^{\textup{ES}} \succeq 0, \forall k \in\{r,t\}.
	\end{align}
\end{subequations}
Then, we consider the Lagrangian function of the objective function as \eqref{L1}, presented at the top of the page. $\psi_k$, $\mu_k$, $\chi_k$, $\zeta$, and $\mathbf{C}$ denote the non-negative Lagrangian coefficients of the corresponding terms. Based on this, the dual of the problem \eqref{ES_Transmit} can be further furmulated as follow:
\begin{subequations}\label{Lagrangian}
	\begin{align}
		&\label{L_0}\min_{\psi_k,\mu_0,\mu_k,\chi_k,\mathbf{C}} \mathcal{L} \\ 
		&\label{L_1}{\rm s.t.}\ \lambda_k \geq 0, \mu_k \geq 0, \chi_k \geq 0, \forall k \in\{t,r\}.\\ 
		&\label{L_2}\quad \quad \zeta \geq 0, \mathbf{C} \succeq 0. 
	\end{align}
\end{subequations}
As can be observed, with given $P_A>0$, a feasible solution can always be found for problem \eqref{ES_Transmit}. In addition to this, the optimal solution is obtained when the constraint $\mathrm{Tr}\left(\mathbf{V}^{\textup{ES}}\right)= P_A$ holds; otherwise, $\mathrm{Tr}\left(\mathbf{V}^{\textup{ES}}\right)$ can be further increased to obtain a larger objective value while keeping the other variables constant. Therefore, the relaxed problem \eqref{ES_Transmit} is feasible and bounded. The same conclusion applies to its dual problem \eqref{Lagrangian}. According to Theorem 3.2 in \cite{Rank_proof}, there always exists an optimal solution satisfying with $\mathrm{Tr}^2\left(\mathbf{V}^{\textup{ES}*}\right)\leq C_n$, where $C_n$ denotes the number of linear constraints associated with $\mathbf{V}^{\textup{ES}*}$ and $C_n=3$ in probelm \eqref{ES_Transmit}. Besides, $\mathbf{V}^{\textup{ES}}=0$ is not the optimal solution. Therefore, $\mathrm{Rank}\left(\mathbf{V}^{\textup{ES}*}\right)=1$ always holds in solving problem \eqref{ES_Transmit}.
\begin{figure*}
	\normalsize
	\begin{align}\label{L1}
		\mathcal{L}&=(1-\psi_t-\psi_r)\Gamma+\psi_tR^{\textup{ES}}_t+\psi_r\widehat{R}^{\textup{ES}}_r-\zeta(\mathrm{Tr}\left(\mathbf{V}^{\textup{ES}}\right)-P_A)+\sum_{k\in\{t,r\}}\chi_kp_k\nonumber \\
		&-\sum_{k\in\{t,r\}}\mu_k\left((1-\overline{\tau})p_k-\overline{\tau}\mathrm{Tr}\left(\!\mathbf{V^{\textup{ES}}}\mathbf{G}_k\widetilde{\mathbf{U}}^{\textup{ES}}_k\mathbf{G}^H_k\right)\right)+\mathrm{Tr}\left(\mathbf{C}\mathbf{V^{\textup{ES}}}\right).
	\end{align}
	\hrulefill \vspace*{0pt}
\end{figure*}
\bibliographystyle{IEEEtran}
\bibliography{mybib.bib}
\end{document}